\documentclass[journal, draftcls, onecolumn]{IEEEtran}

\usepackage{enumerate}
\usepackage{epstopdf}
\usepackage{cite}
\usepackage{amscd}
\usepackage{dsfont}
\usepackage{latexsym}
\usepackage{array}

\usepackage{tikz}\usetikzlibrary{calc}
\usepackage{tabularx}
\usepackage{epsfig}
\usepackage{graphicx,subfigure}
\usepackage{epstopdf}
\usepackage{graphicx}
\usepackage{amsmath,mathrsfs}
\usepackage{amsthm}
\usepackage{amssymb}  
\usepackage{amsbsy}
\usepackage{color}
\usepackage{stfloats}
\usepackage{algorithm}
\usepackage{algpseudocode}
\usepackage{bm}
\usepackage{pifont}

\newtheoremstyle{note}
{3pt}
{1pt}
{}
{\parindent}
{\itshape}
{:}
{.5em}
{\thmname{#1}\thmnumber{ #2}\thmnote{\thmnote{ (#3)}}}
\theoremstyle{note}

\newtheorem{theorem}{Theorem}
\newtheorem{lemma}{Lemma}
\newtheorem{remark}{Remark}
\newtheorem{definition}{Definition}

\theoremstyle{definition}

\newtheoremstyle{dotless}{}{}{\itshape}{}{\bfseries}{}{ }{}
\theoremstyle{dotless}

\newcommand{\R}{\mathbb{R}}

\newcommand {\aplt} {\ {\raise-.5ex\hbox{$\buildrel<\over{\mbox{\scriptsize $\sim$}}$}}\ }
\providecommand{\abs}[1]{\ensuremath{\left\lvert #1 \right\rvert}}

\DeclareMathOperator{\SNR}{\mathsf{SNR}}

\begin{document}

	\title{Construction of Capacity-Achieving Lattice Codes: Polar Lattices}
	\author{Ling Liu, Yanfei~Yan,
	        Cong~Ling,~\IEEEmembership{Member,~IEEE}
	        and Xiaofu Wu,~\IEEEmembership{Member,~IEEE}
    \thanks{This work was presented in part at ITW 2012 and in part at ISIT 2013.}
	\thanks{Ling Liu, Yanfei Yan and Cong Ling are with the Department of Electrical and Electronic Engineering,
	 Imperial College London, London, UK (e-mails: liuling\_88@pku.edu.cn, yanyanfei@gmail.com, cling@ieee.org).}
	\thanks{Xiaofu Wu is with the Nanjing University of Posts and Telecommunications, Nanjing 210003, China (e-mail: xfuwu@ieee.org).}
}

\maketitle

\begin{abstract}
In this paper, we propose a new class of lattices constructed from polar codes, namely polar lattices, to achieve the capacity $\frac{1}{2}\log(1+\SNR)$ of the additive white Gaussian-noise (AWGN) channel. Our construction follows the multilevel approach of Forney \textit{et al.}, where we construct a capacity-achieving polar code on each level. The component polar codes are shown to be naturally nested, thereby fulfilling the requirement of the multilevel lattice construction. We prove that polar lattices are \emph{AWGN-good}. Furthermore, using the technique of source polarization, we propose discrete Gaussian shaping over the polar lattice to satisfy the power constraint. Both the construction and shaping are explicit, and the overall complexity of encoding and decoding is $O(N\log N)$ for any fixed target error probability.

\end{abstract}

\begin{IEEEkeywords}
AWGN-good lattices, discrete Gaussian shaping, lattice codes, multilevel construction, polar codes.
\end{IEEEkeywords}

%
\IEEEpeerreviewmaketitle
%
\section{Introduction}
%
%
%
%

\color{black}A fast-decodable, structured code achieving the capacity of the power-constrained additive white Gaussian-noise (AWGN) channel is a major goal of communication theory. \color{black} Polar codes, proposed by Ar{\i}kan in \cite{polarcodes}, can provably achieve the capacity of binary memoryless symmetric (BMS) channels. An attempt to construct polar codes for the AWGN channel was given in \cite{Abbe}, based on nonbinary polar codes or on the technique for the multi-access channel. Although coded modulation using polar codes has been investigated in literature \cite{multilevelpolar,Mahdavifar16}, the AWGN channel capacity has not been achieved, to the best of our knowledge.

\color{black}Lattice codes are counterparts of linear codes in the Euclidean space. The existence of lattice codes achieving the Gaussian channel capacity has been well established using the random coding argument \cite{zamir,LingBel13}. In the classical point-to-point channel, lattice codes offer a low-complexity solution compared to Gaussian random codes. More recently, thanks to their rich structures, lattice codes have emerged as a novel framework of coding for multiuser communications, such as compute-and-forward \cite{nazer,6582523} and index coding \cite{Natarajan15}. In many problems of Gaussian multiuser networks, lattice codes demonstrate a clear advantage and outperform best known solutions based on random codes. This is because lattice codes enjoy the benefit of coordination despite the distributed nature of coding in a network. Readers are referred to \cite[Chap. 12]{BK:Zamir} for an extensive overview of the applications of lattice codes to Gaussian networks and their advantages over classical random coding approaches.

\color{black}It is well known that the design of a lattice code consists of two essentially separate problems: AWGN coding and shaping. AWGN coding is addressed by the notion of AWGN-good lattices \cite{Poltyrev,zamir}. Recently, several new lattice constructions with good performance have been introduced \cite{ldpclattice,PietroZB18,ldlc}. On the other hand, shaping takes care of the finite power constraint of the Gaussian channel. Capacity-achieving shaping techniques include Voronoi shaping \cite{zamir} and lattice Gaussian shaping \cite{LingBel13,Forney_Wei_89,Kschischang_Pasupathy}. Despite these significant progresses, an explicit construction of lattice codes achieving the capacity of the Gaussian channel is still open (since this work was completed, we have become aware of the work \cite{PietroZB18} which shows low density Construction-A (LDA) lattices achieve capacity when the signal-to-noise ratio (SNR) $>1$ in magnitude).

In this paper, we settle this open problem by employing the powerful tool of polarization in lattice construction. The novel technical contribution of this work is two-fold:
\begin{itemize}
  \item The construction of polar lattices and the proof of their AWGN-goodness. We follow the multilevel construction of Forney, Trott and Chung \cite{forney6}, where for each level we build a polar code to achieve its capacity. We prove that the subchannels arising from some lattice partition chains are successively degraded, which guarantees that the component polar codes are naturally nested, as required by the multilevel construction. This compares favorably with existing multilevel constructions \cite{ldpclattice}, where extra efforts are needed to nest the component codes.
  \item The Gaussian shaping technique for polar lattices in the power-constrained AWGN channel. This is based on source polarization and may be viewed as inverse source coding. Finally, our scheme is able to achieve the capacity $\frac{1}{2}\log(1+\SNR)$ with low-complexity multistage successive cancellation (SC) decoding for any given SNR. It is worth mentioning that our proposed shaping scheme is not only a practical implementation of lattice Gaussian shaping, but also an improvement in the sense that we successfully remove the restriction $\SNR>e$ in \cite[Theorem 3]{LingBel13}.
\end{itemize}

Overall, both source and channel polarization are employed in the construction, resulting in an integrated approach in the sense that error correction and shaping are performed by one single polar code on each level. Moreover, capacity is achieved with minimum mean-square error (MMSE) lattice decoding. The construction of polar lattices with Gaussian shaping is as explicit as that of polar codes themselves, and the complexity is quasilinear: $O(N\log^2 N)$ for a sub-exponentially vanishing error probability and $O(N\log N)$ for a fixed error probability, respectively.

Further, it is worth pointing out that each aspect may also be of independent interest. AWGN-good lattices have many applications in coding and network information theory, while lattice Gaussian shaping, i.e., generating a Gaussian distribution over a lattice, is useful in lattice-based cryptography as well \cite{gaussianlattice1}. Both theoretical and practical aspects of polar lattices are addressed in this work. We not only prove the theoretical goodness of polar lattices, but also give practical rules for designing these lattices.

\subsection{Relation to Prior Works}

This paper is built on the basis of our prior attempt to build lattices from polar codes \cite{Yan,yan2}, and significantly extends it by employing Gaussian shaping. We are aware of the contemporary and independent work on polar-coded modulation \cite{multilevelpolar}, which follows the multilevel coding approach of \cite{multilevel}. It is known that Forney \textit{et al.}'s multilevel construction is closely related to multilevel coding \cite{multilevel,forney6}. The main conceptual difference between lattice coding and coded modulation is that lattices are infinite and linear in the Euclidean space. The linear structure of lattices is much desired in many emerging applications, e.g., in network information theory for the purpose of coordination \cite{zamir1,nazer}.

This paper may be viewed as an explicit construction of the lattice Gaussian coding scheme proposed in \cite{LingBel13}, where it was shown that Gaussian shaping over an AWGN-good lattice is capacity-achieving. Our approach is different from the standard Voronoi shaping which involves a quantization-good lattice \cite{zamir}. The proposed Gaussian shaping does not require such a quantization-good lattice any more.

The sparse superposition code \cite{sparsecode1,sparsecode2} also achieves the Gaussian channel capacity with polynomial complexity. However, its decoding complexity is considerably higher than that of the polar lattice; moreover, it requires a random dictionary shared by the encoder and decoder, which incurs substantial storage complexity. In comparison, the construction of polar lattices is as explicit as that of polar codes themselves, and the complexity is quasilinear: $O(N\log^2 N)$ for a sub-exponentially vanishing error probability and $O(N\log N)$ for a fixed error probability, respectively.

\subsection{Organization and Notation}

The rest of this paper is organized as follows.
Section II presents the background of lattice codes. In Section III, we construct polar latices based on Forney \textit{et al.}'s approach and prove their AWGN-goodness. In Section IV, we propose Gaussian shaping over the polar lattice to achieve the capacity. Section V gives design examples and simulation results. Section VI concludes the paper.


All random variables (RVs) will be denoted by capital letters. For a set $\mathcal{I}$, $\mathcal{I}^c$ denotes its complement, and $|\mathcal{I}|$ represents its cardinality. Following the notation of \cite{polarcodes}, we denote $N$ independent uses of channel $W$ by $W^N$. By channel combining and splitting, we get the combined channel $W_N$ and the $i$-th subchannel $W_N^{(i)}$. Throughout this paper, we use the binary logarithm, denoted by log, and information is measured in bits. We follow the standard asymptotic notation $%
f\left( x\right) =O\left( g\left( x\right) \right) $ if $\lim\sup_{x\rightarrow
\infty}|f(x)/g(x)| < \infty$.

\section{Background on Lattice Coding}

A lattice is a discrete subgroup of $\mathbb{R}^{n}$ which can be represented by
\begin{eqnarray}
\Lambda=\{ \lambda={B}{x}:{x}\in\mathbb{Z}^{n}\}, \notag\
\end{eqnarray}
where the generator matrix ${B}$ is assumed to be of full rank in this paper. The theta series of $\Lambda$ is defined as
\begin{eqnarray}
\Theta_{\Lambda}(\tau)=\sum_{\lambda\in\Lambda}e^{-\pi\tau\parallel\lambda\parallel^{2}}, \quad \tau>0.
\notag\
\end{eqnarray}
Readers are referred to the text \cite{BK:Zamir} for basic definitions of lattices.

In this work, we are mostly concerned with the block error probability $P_{e}(\Lambda,\sigma^2)$ of lattice decoding. It is the probability $\mathbb{P}\{{x}\notin \mathcal{V}(\Lambda)\}$ that an $n$-dimensional independent and identically distributed (i.i.d.) Gaussian noise vector ${x}$ with zero mean and variance $\sigma^{2}$ per dimension falls outside the Voronoi region $\mathcal{V}(\Lambda)$. For an $n$-dimensional lattice $\Lambda$, the volume of a fundamental region is given by $V(\Lambda)=|\text{det}({B})|$. Define the VNR by
\begin{eqnarray}
\gamma_{\Lambda}(\sigma)\triangleq\frac{V(\Lambda)^\frac{2}{n}}{\sigma^2}. \notag\
\end{eqnarray}
A sequence of lattices $\Lambda^{(N)}$ of increasing dimension $N$ is AWGN-good if, for any fixed VNR greater than $2\pi e$, \[
\lim_{N\rightarrow\infty} P_e(\Lambda^{(N)},\sigma^2) = 0.
\]
It is worth mentioning here that we do not insist on exponentially vanishing error probabilities, unlike Poltyrev's original treatment of good lattices for coding over the AWGN channel \cite{Poltyrev}. This is because a sub-exponential or polynomial decay of the error probability is often good enough.

For $\sigma>0$ and ${c}\in\mathbb{R}^{n}$, the Gaussian distribution of mean ${c}$ and variance $\sigma^{2}$ is defined as
\begin{eqnarray}
f_{\sigma,{c}}({x})=\frac{1}{(\sqrt{2\pi}\sigma)^{n}}e^{-\frac{\parallel {x}-{c}\parallel^{2}}{2\sigma^{2}}}, \notag\
\end{eqnarray}
for all ${x}\in\mathbb{R}^{n}$. For convenience, let $f_{\sigma}({x})=f_{\sigma,{0}}({x})$.

Given a lattice $\Lambda$, we define the $\Lambda$-periodic function as
\begin{eqnarray}
f_{\sigma,\Lambda}({x})=\sum\limits_{\lambda\in\Lambda}f_{\sigma,\lambda}({x})=\frac{1}{(\sqrt{2\pi}\sigma)^{n}}\sum\limits_{\lambda\in\Lambda}e^{-\frac{\parallel {x}-\lambda\parallel^{2}}{2\sigma^{2}}}, \notag\
\end{eqnarray}
for ${x}\in\mathbb{R}^n$.
Note that $f_{\sigma,\Lambda}({x})$ is a probability density if ${x}$ is restricted to a fundamental region $\mathcal{R}(\Lambda)$. It is actually the probability density function (PDF) of the $\Lambda$-aliased Gaussian noise, i.e., the Gaussian noise after the mod-$\mathcal{R}(\Lambda)$ operation \cite{forney6}. When $\sigma$ is small, the effect of aliasing becomes insignificant and the $\Lambda$-aliased Gaussian density $f_{\sigma,\Lambda}({x})$ approaches a Gaussian distribution. When $\sigma$ is large, $f_{\sigma,\Lambda}({x})$ approaches a uniform distribution.
This phenomenon is characterized by the flatness factor, which is defined for $\Lambda$ as \cite{cong2}
\begin{eqnarray}
\epsilon_{\Lambda}(\sigma)\triangleq\max\limits_{{x}\in\mathcal{R}(\Lambda)}\abs{V(\Lambda)f_{\sigma,\Lambda}({x})-1}. \notag\
\end{eqnarray}
It can be interpreted as the maximum variation of $f_{\sigma,\Lambda}({x})$ from the uniform distribution over $\mathcal{R}(\Lambda)$.

We define the \emph{discrete Gaussian distribution} over $\Lambda$
centered at ${c} \in \R^n$ as the following discrete
distribution taking values in ${ \lambda} \in \Lambda$:
\[
D_{\Lambda,\sigma,{c}}({ \lambda})=\frac{f_{\sigma,{c}}({{ \lambda}})}{f_{\sigma,{c}}(\Lambda)}, \quad \forall { \lambda} \in \Lambda,
\]
where $f_{\sigma,{c}}(\Lambda) \triangleq \sum_{{ \lambda} \in
\Lambda} f_{\sigma,{c}}({{ \lambda}})=f_{\sigma,\Lambda}({c})$. Again for convenience, we write $D_{\Lambda,\sigma}=D_{\Lambda,\sigma,{0}}$. Figure \ref{fig:Discrete_Gaussian} illustrates the discrete Gaussian distribution over $\mathbb{Z}^2$. As can be seen, it resembles a continuous Gaussian distribution, but is only defined over a lattice. In fact, discrete and continuous Gaussian distributions share similar properties, if the flatness factor is small.
The discrete Gaussian distribution can also be sampled from a shifted lattice $\Lambda-{c}$:
\[
D_{\Lambda-{c},\sigma}({ \lambda}-{c})=\frac{f_{\sigma}({{ \lambda}}-{c})}{f_{\sigma, {\bf c}}(\Lambda)}, \quad \forall { \lambda} \in \Lambda.
\]
Note the relation $D_{\Lambda-{c},\sigma}({ \lambda}-{c}) = D_{\Lambda,\sigma,{c}}({ \lambda})$, namely, they are a shifted version of each other.

\begin{figure}[t]

\centering\centerline{\epsfig{figure=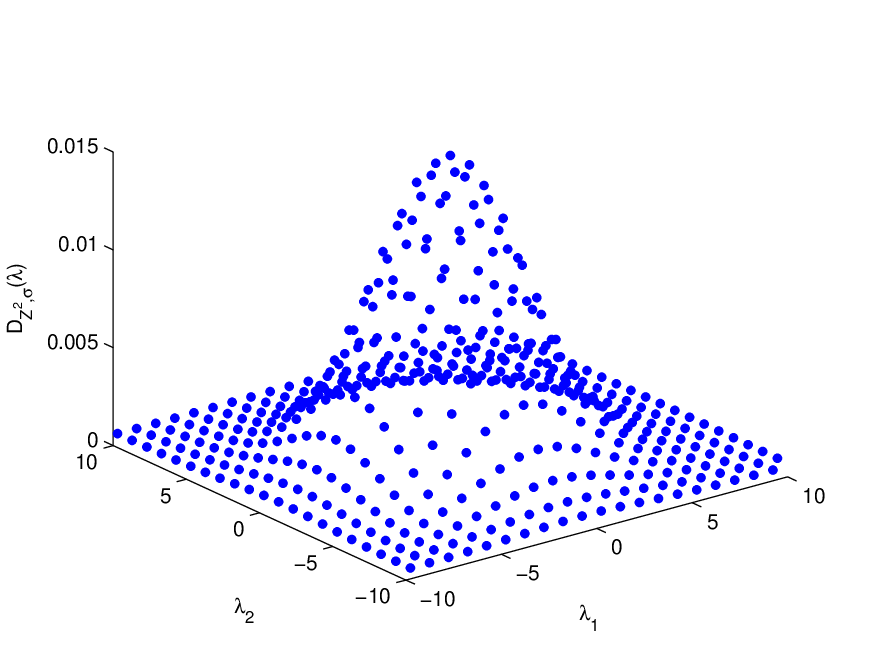,width=8cm}}

\caption{Discrete Gaussian distribution over $\mathbb{Z}^2$. A two-dimensional lattice point is denoted by $\lambda=(\lambda_1,\lambda_2)$.}

\vspace{-0.5cm}

\label{fig:Discrete_Gaussian}
\end{figure}

If the flatness factor is negligible, the discrete
Gaussian distribution over a lattice preserves the capacity of the AWGN channel \cite[Theorem 2]{LingBel13}.

\begin{theorem}[Mutual information of discrete Gaussian distribution \cite{LingBel13}]\label{theorem:capacity}
Consider an AWGN channel ${Y}={X}+{E}$ where the input constellation $X$ has
a discrete Gaussian distribution $D_{\Lambda-{c},\sigma_s}$
for arbitrary ${c} \in \mathbb{R}^n$, and where the variance
of the noise $E$ is $\sigma^2$. Let the average signal power be $P$ so that $\SNR=P/\sigma^2$, and let $\tilde{\sigma}\triangleq \frac{\sigma_s\sigma}{\sqrt{\sigma_s^2+\sigma^2}}$. Then, if
$\varepsilon = \epsilon_{\Lambda}\left(\tilde{\sigma}\right) < \frac{1}{2}$ and $\frac{\pi\varepsilon_t}{1-\epsilon_t}\leq \varepsilon$ where
\[
\varepsilon_t \triangleq
\left\{
  \begin{array}{ll}
    \epsilon_{\Lambda}\left(\sigma_s/\sqrt{\frac{\pi}{\pi-t}}\right), & \hbox{$t \geq 1/e$} \\
    (t^{-4}+1)\epsilon_{\Lambda}\left(\sigma_s/\sqrt{\frac{\pi}{\pi-t}}\right), & \hbox{$0< t < 1/e$}
  \end{array}
\right.
\]
the discrete Gaussian constellation results in mutual information
\begin{equation}\label{eq:lattice-capacity}
I_D \geq \frac{1}{2}\log {(1+\SNR)} - \frac{6\varepsilon}{n}
\end{equation}
per channel use.
\end{theorem}

The statement of Theorem \ref{theorem:capacity} is non-asymptotical, i.e., it can hold even if $n=1$. A lattice $\Lambda$ or its coset $\Lambda-{c}$ with a discrete Gaussian distribution is referred to as a \emph{good constellation} for the AWGN channel if ${\epsilon_{\Lambda}(\tilde{\sigma})}$ is negligible \cite{LingBel13}.

It is further proved in \cite{LingBel13} that the channel capacity is achieved with Gaussian shaping over an AWGN-good lattice and MMSE lattice decoding. To this aim, we use a codebook $L-{c}$, where $L$ is an AWGN-good lattice and ${c}$ is a proper shift. The encoder maps the information bits to points in $L-{c}$, which obey the lattice Gaussian distribution $D_{L-{c},\sigma_s}$. Since the lattice points are not equally probable a priori in the lattice Gaussian coding, we apply maximum-a-posteriori (MAP) decoding. It is proved in \cite{LingBel13} that MAP decoding is equivalent to MMSE lattice decoding
\begin{equation}\label{MAPdecoder}
\hat{{x}}=Q_{L-{c}}\left(\alpha{{y}}\right)
\end{equation}
where $\alpha=\frac{\sigma_s^2}{\sigma_s^2+\sigma^2}$ is asymptotically equal to the MMSE coefficient $\frac{P}{P+\sigma^2}$ and $Q_{L-{c}}$ denotes the minimum Euclidean-distance decoder for shifted lattice $L-{c}$.

\section{Construction of Polar Lattices}
\label{sec:PL}

We now follow Forney {\it et al.}'s multilevel approach \cite{forney6} to construct polar lattices.
Bear in mind that, in order to achieve the capacity of the AWGN channel with the noise variance $\sigma^2$, the concerned noise variance for the AWGN-good lattice is in fact $\tilde{\sigma}^2$ (recall $\tilde{\sigma}\triangleq \frac{\sigma_s\sigma}{\sqrt{\sigma_s^2+\sigma^2}}$), which is the variance of the equivalent noise after MMSE rescaling \cite{LingBel13}.

\subsection{Forney \textit{et al.}'s Construction}

Given a sublattice $\Lambda' \subset \Lambda$, the quotient group $\Lambda/\Lambda'$ induces a partition of $\Lambda$ into equivalence classes modulo $\Lambda'$. We call $\Lambda/\Lambda'$ a lattice partition \cite{forney6}. The order of the partition is denoted by $|\Lambda/\Lambda'|$, which is equal to the number of cosets. If $|\Lambda/\Lambda'|=2$, we call this a binary partition. Similarly, if $\Lambda' \subseteq \Lambda_{r-1} \subseteq \cdots \subseteq\Lambda_{1} \subseteq \Lambda$ for $r \geq 1$ is a chain of lattices with quotients $\Lambda/\Lambda_{1}/\cdots/\Lambda_{r-1}/\Lambda'$, then $\Lambda/\Lambda_{1}/\cdots/\Lambda_{r-1}/\Lambda'$ is called an $n$-dimensional lattice partition chain. For each
partition $\Lambda_{\ell-1}/\Lambda_{\ell}$ ($1\leq \ell \leq r$ with convention $\Lambda_0=\Lambda$ and $\Lambda_r=\Lambda'$), a code $\mathcal{C}_{\ell}$ over $\Lambda_{\ell-1}/\Lambda_{\ell}$
selects a sequence of representatives $a_{\ell}$ for the cosets of $\Lambda_{\ell}$. Consequently, if each partition is a binary partition, the codes $\mathcal{C}_{\ell}$ are binary codes.

Construction D requires a set of nested linear binary codes $\mathcal{C}_{1}\subseteq\mathcal{C}_{2}\cdot\cdot\cdot\subseteq\mathcal{C}_{r}$ \cite{forney6}.
Suppose $\mathcal{C}_{\ell}$ has block length $N$ and the number of information bits $k_{\ell}$  for $1\leq\ell\leq r$. Choose a basis $\mathbf{g}_{1}, \mathbf{g}_{2},\cdots, \mathbf{g}_{N}$  such that $\mathbf{g}_{1},\cdots \mathbf{g}_{k_{\ell}}$ span $\mathcal{C}_{\ell}$. In this work, we focus on the one-dimensional partition chain $\mathbb{Z}/2\mathbb{Z}/\cdot\cdot\cdot/2^{r}\mathbb{Z}$ for the simplicity of presentation. Accordingly, the lattice $L$ admits the form \cite{forney6}
\begin{eqnarray}
L = \left\{\sum_{\ell=1}^{r}2^{\ell-1}\sum_{i=1}^{k_{\ell}}u_{\ell}^{i}\mathbf{g}_{i}+2^{r}\mathbb{Z}^N\mid u_{\ell}^{i}\in\{0,1\}\right\}
\label{constructionD}
\end{eqnarray}
where the addition is carried out in $\mathbb{R}^N$.
The fundamental volume of a lattice obtained from this construction is given by
\begin{eqnarray}
V(L)=2^{-NR_{\mathcal{C}}}V(\Lambda')^{N}, \notag\
\end{eqnarray}
where $R_{\mathcal{C}}=\sum_{\ell=1}^{r} R_{\ell}=\frac{1}{N}\sum_{\ell=1}^{r}k_{\ell}$ denotes the sum rate of component codes.

The following is an example of Construction D: Barnes-Wall lattices constructed from Reed-Muller codes \cite{forney1}. We give the example of Barnes-Wall lattices as a benchmark particularly because of the connection between Reed-Muller codes and polar codes \cite{polarcodes}. The advantage of polar codes over Reed-Muller codes will translate into the advantage of polar lattices over Barnes-Wall lattices.
Reed-Muller codes RM$(N, k, d)$ are a class of linear block codes
over GF$(2)$, where $N$ is the length of the codeword, $k$ is the
length of the information block and $d$ is the minimum Hamming
distance. \color{black}Conventionally, Reed-Muller codes are
denoted by RM$(r', m)$ $(0\leq r'\leq m)$ with following relation
among $N$, $k$ and $d$:
\begin{eqnarray}
N=2^m, k=1+\binom{m}{1}+\cdots+\binom{m}{r'}, d=2^{m-r'}. \notag\
\label{eqn:outputandinput}
\end{eqnarray}

The $m$-th member of the family of Barnes-Wall lattices is an $N=2^m$ dimensional complex lattice or $2N$ dimensional real lattice. \color{black}For example, the code formula of the $1024$-dimensional Barnes-Wall lattice is:
\begin{eqnarray}
BW_{1024}=\text{RM}(1,10)+2\text{RM}(3,10)+\cdot\cdot\cdot+2^{5}\mathbb{Z}^{1024}.
\label{eqn:BW1024}
\end{eqnarray}

A mod-$\Lambda$ Gaussian channel is a Gaussian channel with an input in $\mathcal{V}(\Lambda)$ and with a mod-$\mathcal{V}(\Lambda)$ operator at the receiver front end \cite{forney6}. The capacity of the mod-$\Lambda$ channel for noise variance $\sigma^2$ is
\begin{eqnarray}\label{mod-capacity}
C(\Lambda, \sigma^{2})=\log  V(\Lambda)-h(\Lambda, \sigma^{2}),
\label{eqn:modchannelcapacity}
\end{eqnarray}
where $h(\Lambda, \sigma^{2})$ is the differential entropy of the $\Lambda$-aliased noise over $\mathcal{V}(\Lambda)$:
\begin{eqnarray}
h(\Lambda,\sigma^{2})&=&-\int_{\mathcal{V}(\Lambda)}f_{\sigma,\Lambda}({x})\text{ log } f_{\sigma,\Lambda}({x})d{x}. \notag\
\end{eqnarray}

Given lattice partition $\Lambda/\Lambda'$, the $\Lambda/\Lambda'$ channel is a mod-$\Lambda'$ channel whose input is restricted to discrete lattice points in $(\Lambda + a) \cap \mathcal{R}(\Lambda')$ for some translate $a$.
The capacity of the $\Lambda/\Lambda'$ channel is given by \cite{forney6}
\begin{equation}\label{mod12-capacity}
\begin{split}
  C(\Lambda/\Lambda', \sigma^2) &= C(\Lambda', \sigma^2) - C(\Lambda, \sigma^2) \\
  &= h(\Lambda, \sigma^2) - h(\Lambda', \sigma^2) + \log (V(\Lambda')/V(\Lambda)).
\end{split}
\end{equation}
Further, if $\Lambda/\Lambda_{1}/\cdots/\Lambda_{r-1}/\Lambda'$ is a lattice partition chain, then
\begin{equation}
C(\Lambda/\Lambda', \sigma^2) = C(\Lambda/\Lambda_1, \sigma^2) + \cdots + C(\Lambda_{r-1}/\Lambda', \sigma^2).
\end{equation}

The key idea of \cite{forney6} is to use a good component code $\mathcal{C}_{\ell}$ to achieve the capacity $C(\Lambda_{\ell-1}/\Lambda_{\ell}, \sigma^2)$ for each level $\ell=1,2,\ldots,r$ in Construction D. For such a construction, the total decoding error probability with multistage decoding is bounded by
\begin{equation}\label{total-Pe}
    P_e(L, \sigma^{2}) \leq \sum_{\ell=1}^{r}{ P_e(\mathcal{C}_{\ell},\sigma^{2})} + P_e((\Lambda')^{N},\sigma^{2}).
\end{equation}
To achieve a vanishing error probability, i.e., to make $P_e(L, \sigma^{2}) \to 0$, we need to choose the lattice $\Lambda'$ such that $P_e((\Lambda')^{N},\sigma^{2}) \to 0$ and that all the codes $\mathcal{C}_{\ell}$ for the $\Lambda_{\ell-1}/\Lambda_{\ell}$ channels have error probabilities tending to zero.

Since $V(L)=2^{-NR_{\mathcal{C}}}V(\Lambda')^{N}$, the logarithmic VNR of $L$ is
\begin{eqnarray}
\log \left(\frac{\gamma_{L}(\sigma)}{2\pi e}\right)
&=&\log\frac{V(L)^\frac{2}{nN}}{2\pi e\sigma^2} \notag\\
&=&\log\frac{2^{-\frac{2}{n}R_{\mathcal{C}}}V(\Lambda')^{\frac{2}{n}}}{2\pi e\sigma^2} \notag\\
&=&-\frac{2}{n}R_{\mathcal{C}}+\frac{2}{n}\log V(\Lambda')- \log2\pi e\sigma^2.
\label{eqn:VNRwithe}
\end{eqnarray}

Define
\begin{equation}
\begin{cases} \epsilon_{1}=C(\Lambda,\sigma^2) \\
\epsilon_{2}=h(\sigma^2)-h(\Lambda',\sigma^2) \\
\epsilon_{3}=C(\Lambda/\Lambda', \sigma^{2})-R_{\mathcal{C}}=\sum_{\ell=1}^{r}{C(\Lambda_{\ell-1}/\Lambda_{\ell}, \sigma^{2})-R_{\ell}},
\end{cases}
\label{eqn:epsilons}
\end{equation}
where $h(\sigma^2)=\frac{n}{2}\log2\pi e\sigma^2$ is the differential entropy of the Gaussian noise. We note that, $\epsilon_{1}\geq0$ represents the capacity of the mod-$\Lambda$ channel, $\epsilon_{2}\geq0$ (due to the data processing inequality) is the difference between the entropy of the Gaussian noise and that of the mod-$\Lambda'$ Gaussian noise, and $\epsilon_{3}\geq0$ is the total capacity loss of component codes.

Then we have
\begin{eqnarray}
\log\left(\frac{\gamma_{L}(\sigma)}{2\pi e}\right)
=\frac{2}{n}(\epsilon_{1}-\epsilon_{2}+\epsilon_{3}). \notag\
\end{eqnarray}
Since $\epsilon_2 \geq 0$, we obtain the upper bound\footnote{It was shown in \cite{forney6} that $\epsilon_{2}\approx\pi P_e(\Lambda',\sigma^2)$, which is negligible compared to the other two terms.}
\begin{eqnarray}
\log\left(\frac{\gamma_{L}(\sigma)}{2\pi e}\right)
\leq \frac{2}{n}(\epsilon_{1}+\epsilon_{3}).
\label{eqn:minimumVNR}
\end{eqnarray}
Since $\log\left(\frac{\gamma_{L}(\sigma)}{2\pi e}\right)=0$ represents the Poltyrev capacity \cite{Poltyrev}\cite[Theorem 6.3.1]{BK:Zamir}, i.e., the capacity per unit volume of an unconstrained AWGN channel, the right hand side of \eqref{eqn:minimumVNR} gives an upper bound on the gap to the Poltyrev capacity. The bound is equal to $\frac{6.02}{n}(\epsilon_1+\epsilon_3)$ decibels (dB), by conversion of the binary logarithm into the base-$10$ logarithm.

To approach the Poltyrev capacity, we would like to have $P_e(L, \sigma^2) \to 0$ for any $\log\left(\frac{\gamma_{L}(\sigma)}{2\pi e}\right) > 0$. Thus, from (\ref{eqn:minimumVNR}), we need that both ${\epsilon_1}$ and ${\epsilon_3}$ are arbitrarily small.
In the following lemma, we upper-bound ${\epsilon_1}$ by the flatness factor $\epsilon_{\Lambda}(\sigma)$ of the top lattice.

\begin{lemma}\label{lem:mod-capacity}
The capacity $C(\Lambda, \sigma^2)$ of the mod-$\Lambda$ channel is bounded by
\begin{equation}\label{capacity1-flatness}
    C(\Lambda, \sigma^2) \leq \log{(1+ \epsilon_{\Lambda}(\sigma))} \leq \log(e) \cdot \epsilon_{\Lambda}(\sigma).
\end{equation}
\end{lemma}

\begin{proof}
By the definition of the flatness factor, we have
\[
f_{\sigma,\Lambda}({x}) \leq  \frac{1+ \epsilon_{\Lambda}(\sigma)}{V(\Lambda)}.
\]
Thus, the differential entropy of the mod-$\Lambda$ Gaussian noise is bounded by
\begin{equation}\label{mod-capacity-flatness}
\begin{split}
  h(\Lambda, \sigma^{2})
  &=-\int_{\mathcal{V}(\Lambda_{1})}f_{\sigma,\Lambda}({x})\text{ log } f_{\sigma,\Lambda}({x})d{x}\\
  &\geq  -\int_{\mathcal{V}(\Lambda_{1})}f_{\sigma,\Lambda}({x})\text{ log } \frac{1+ \epsilon_{\Lambda}(\sigma)}{V(\Lambda)} d{x}\\
  &= -\log \frac{1+ \epsilon_{\Lambda}(\sigma)}{V(\Lambda)} \\
  &= \log V(\Lambda) - \log{(1+ \epsilon_{\Lambda}(\sigma))}.\notag\
\end{split}
\end{equation}
Therefore, from \eqref{eqn:modchannelcapacity}, $C(\Lambda, \sigma^{2})$ is bounded by $\log{(1+ \epsilon_{\Lambda}(\sigma))}$. The second inequality in (\ref{capacity1-flatness}) follows from the fact $\log(1+x) = \log_2(e) \cdot \log_e(1+x) \leq \log(e) \cdot x$ for $x>0$.
\end{proof}



Thus, we have the following design criteria:
\begin{itemize}
  \item The top lattice $\Lambda$ has a negligible flatness factor $\epsilon_{\Lambda}(\sigma)$.
  \item The bottom lattice $\Lambda'$ has a small error probability $P_e(\Lambda',\sigma^{2})$.
  \item Each component code $\mathcal{C}_{\ell}$ is a capacity-approaching code for the $\Lambda_{\ell-1} / \Lambda_{\ell}$ channel.
\end{itemize}


Asymptotically, the error probability of a polar code of codeword length $N$ decreases approximately as $O(2^{-\sqrt{N}})$ \cite{polarcodes1} and we may desire a similar form for the error probability of a polar lattice. In \eqref{total-Pe}, we can let $P_e((\Lambda')^{N},\sigma^{2})$ decrease exponentially by increasing the volume of the bottom lattice $\Lambda'$ or equivalently by expanding the partition chain. More explicitly, the next lemma shows that the first two ceriteria can be satisfied by $r$ growing with $\log N$ (see Appendix \ref{appendixlevels} for a proof).

\begin{lemma}\label{lem:numberoflevels}
Consider a partition chain $\Lambda/\Lambda_{1}/\cdots/\Lambda_{r-1}/\Lambda'$. There exists a sequence of numbers of levels $r= O(\log N)$ such that $\epsilon_{\Lambda}(\sigma) = O(e^{-{N}})$ and $P_{e}(\Lambda',\sigma^{2}) = O(e^{-{N}})$.
\end{lemma}

\begin{remark}
Lemma \ref{lem:numberoflevels} is mostly of theoretical interest, e.g., for proving a partition chain with increasing levels is capacity achieving. In practical designs, if the target error probability is fixed, e.g., $P_e(L,\sigma^2)=10^{-5}$, a small number of levels will suffice. This is because one can choose a top lattice such that $\epsilon_{\Lambda}(\sigma) \approx 10^{-2}$ and a bottom lattice such that $P_{e}(\Lambda',\sigma^{2}) \approx 10^{-6}$, for instance. In fact, it was shown in \cite{forney6} that a two-level partition chain $\mathbb{Z}/2\mathbb{Z}/4\mathbb{Z}$ is enough if $n=1$, although slightly more levels are needed if $n>1$. Readers are referred to \cite{forney6} for more details and Section \ref{sec:practicaldesign} for design examples.
\end{remark}

\subsection{Polar Lattices}
\label{sec:polarlattice}

It is shown in \cite{forney6} that the $\Lambda_{\ell-1}/\Lambda_{\ell}$ channel is symmetric, and the optimum input distribution is uniform \cite{forney6}. Since we use a binary partition $\Lambda_{\ell-1}/\Lambda_{\ell}$, the input $X_{\ell}$ is binary for $\ell \in \{1,2,\ldots,r\}$. Associate $X_{\ell}$ with representative $a_{\ell}$ of the coset in the quotient group $\Lambda_{\ell-1}/\Lambda_{\ell}$. The fact that the $\Lambda_{\ell-1}/\Lambda_{\ell}$ channel is a BMS channel allows a polar code to achieve its capacity.

Let $Y$ denote the output of the AWGN channel. Given $x_{1:\ell-1}$, let $\mathcal{A}_\ell(x_{1:\ell})$ denote the coset chosen by $x_{\ell}$, i.e., $\mathcal{A}_{\ell}(x_{1:\ell})=a_1+\cdots+a_{\ell}+\Lambda_{\ell}$. The conditional PDF of this channel with input $x_\ell$ and output $\bar{y}_\ell =y\text{ mod }\Lambda_{\ell}$ is given by\cite{forney6}
\begin{eqnarray}
P_{\bar{Y}_\ell|X_\ell,X_{1:\ell-1}}(\bar{y}_\ell|x_\ell,x_{1:\ell-1})&=&f_{{\sigma},\Lambda_{\ell}}(\bar{y}_\ell-a_1-\cdots-a_{\ell}) \notag \\
&=&\frac{1}{\sqrt{2\pi}\sigma}\sum\limits_{a\in \mathcal{A}_\ell(x_{1:\ell})}\text{exp}\left(-\frac{\|\bar{y}_\ell-a\|^2}{2\sigma^{2}}\right).
\label{eqn:modchannel}
\end{eqnarray}

\begin{definition}(Channel degradation \cite{BK:Cover})\label{deft:degradation}: Consider two channels $W_1:\mathcal{X}\rightarrow\mathcal{Y}_1$ and $W_2:\mathcal{X}\rightarrow\mathcal{Y}_2$. Channel $W_1$ is said to be (stochastically) degraded with respect to $W_2$ if there exists a channel $Q: \mathcal{Y}_2\rightarrow\mathcal{Y}_1$ such that
\begin{eqnarray}
W_1(y_1|x)=\sum_{y_2\in\mathcal{Y}_2}W_2(y_2|x)Q(y_1|y_2).\notag\
\end{eqnarray}
\end{definition}

The proof of the following lemma is given in Appendix \ref{Appendixdegraded}.

\begin{lemma}\label{lemma:degraded}
Consider a self-similar binary lattice partition chain $\Lambda/\Lambda_{1}/\cdots/\Lambda_{r-1}/\Lambda'$, in which we have $\Lambda_\ell=T^\ell \Lambda$ for all $\ell$, with $T=\alpha V$ for some scale factor $\alpha>1$ and orthogonal matrix $V$. Then, the $\Lambda_{\ell-1}/\Lambda_{\ell}$ channel is degraded with respect to the $\Lambda_{\ell}/\Lambda_{\ell+1}$ channel for $1\leq \ell \leq r-1$.
\end{lemma}

Now, we recall some basics of polar codes.
Let $W(y|x)$ be a BMS channel with input alphabet $\mathcal{X}=\{0,1\}$ with a priori distribution $\mathrm{Bernoulli}(1/2)$ and output alphabet $\mathcal{Y}\subseteq\mathbb{R}$. Polar codes are block codes of length $N=2^m$ with input bits $u^{1:N}$. Let $I(W)$ be the capacity of $W$. Given a rate $R<I(W)$, the information bits are indexed by a set of $RN$ rows of the generator matrix $G_N=\left[\begin{smallmatrix}1&0\\1&1\end{smallmatrix}\right]^{\otimes m}$, where $\otimes$ denotes the Kronecker product. This gives an $N$-dimensional channel $W_{N}(y^{1:N}|u^{1:N})$. The channel seen by each bit \cite{polarcodes} is given by
\begin{eqnarray}
W_{N}^{(i)}(y^{1:N},u^{1:i-1}|u^{i})=\sum\limits_{u^{i+1:N}\in \mathcal{X}^{N-i}}\frac{1}{2^{N-1}}W_{N}(y^{1:N}|u^{1:N}). \notag\
\end{eqnarray}
Ar{\i}kan proved that as $N$ grows, each channel $W_{N}^{(i)}$ approaches either an error-free channel or a completely noisy channel. The set of almost completely noisy (resp. almost error-free) subchannels is called the frozen set $\mathcal{F}$ (resp. information set $\mathcal{I}$). One sets $u^{i}=0$ for $i\in \mathcal{F}$ and only sends information bits within $\mathcal{I}$.

Given a priori input distribution $\mathrm{Bernoulli}(1/2)$, the error probability of channel $W$ with transition probability $P_{Y|X}$ under maximum-likelihood
decision is given by
\begin{eqnarray}
P_e(W) = \frac{1}{2}\sum\limits_{y}{\min\{P_{Y|X}(y|0),P_{Y|X}(y|1)\}}. \notag\
\end{eqnarray}
The Bhattacharyya parameter serves as an upper bound on $P_e(W)$.
\begin{definition}[Bhattacharyya Parameter for Symmetric Channel \cite{polarcodes}]\label{deft:symZ}
Given a BMS channel $W$ with transition probability $P_{Y|X}$, the Bhattacharyya parameter $Z\in[0,1]$ is defined as
\begin{eqnarray}
Z(W)&\triangleq\sum\limits_{y} \sqrt{P_{Y|X}(y|0)P_{Y|X}(y|1)}. \notag\
\end{eqnarray}
\end{definition}

The rule of SC decoding is defined as
\begin{eqnarray}
\hat{u}^{i}=\left\{\begin{aligned}
&0\:\:\:\:\:\: i\in \mathcal{F}\:\:\:\text{ or } \:\:\:\ \frac{W_{N}^{(i)}(y^{1:N},\hat{u}^{1:i-1}|0)}{W_{N}^{(i)}(y^{1:N},\hat{u}^{1:i-1}|1)}\geq1 \:\:\text{ when } i \in \mathcal{I}, \\ \notag\
&1\:\:\:\:\:\: \text{otherwise}.
\end{aligned}\right. \notag\
\end{eqnarray}

Let $P_B$ denote the block error probability of a binary polar code under SC decoding.
It has been proved in \cite{polarcodes} that $P_{B}$ can be upper-bounded by the sum of the decoding error probability of the genie-aided SC decoder for each information bit, i.e., $P_{B}\leq\Sigma_{i\in \mathcal{I}}Z(W_{N}^{(i)})$. It is worth mentioning that there are some other decoding methods such as belief propagation decoding \cite{PolarBP} and list decoding \cite{ListPolar}, which perform better than SC decoding. However, in this work, we focus on SC decoding because it is sufficient to show that polar lattices are able to achieve the capacity of AWGN channels.

It was shown in \cite{polarcodes1,polarchannelandsource} that for any $\beta<\frac{1}{2}$,
\begin{eqnarray}
\lim_{m\rightarrow \infty}\frac{1}{N}\left|\{i:Z(W_{N}^{(i)})<2^{-N^{\beta}}\}\right|&=&I(W) \notag \\
\lim_{m\rightarrow \infty}\frac{1}{N}\left|\{i:I(W_{N}^{(i)})>1-2^{-N^{\beta}}\}\right|&=&I(W). \notag
\end{eqnarray}
This means that the fraction of good channels approaches to $I(W)$ as $m\rightarrow \infty$. Therefore, constructing polar codes is equivalent
to choosing the good indices.

Let $\mathcal{P}(N,k_{\ell})$ denote the component polar code for the $\Lambda_{\ell-1}/\Lambda_{\ell}$ partition channel ($1\leq\ell\leq r$), where $k_\ell$ is the size of its information set and $N$ is the block length. We stack them as in Construction D to build the polar lattice.
The following lemma shows that these component codes are nested, which guarantees that the multilevel construction creates a lattice \cite{forney6}.
Two rules may be used to determine the component codes. One is the \emph{capacity rule} \cite{forney6,multilevel}, where the channel indices are selected according to a threshold on the mutual information. The other is the \emph{equal-error-probability rule} \cite{multilevel}, namely, the same error probability for each level, where we select the channel indices according to a threshold on the error probability or the Bhattacharyya parameter. The advantage of the equal-error-probability rule is that it gives an estimate of the error probability. For this reason, we use the equal-error-probability rule in this paper. It is well known that the polar codes constructed according to these two rules converge to each other as the block length goes to infinity \cite{polarcodes}.

\begin{lemma}\label{lem:nested}
For the equal-error-probability rule based on either the error probability or the Bhattacharyya parameter, the component polar codes built in the multilevel construction are nested, i.e., $\mathcal{P}(N,k_{1})\subseteq \mathcal{P}(N,k_{2})\subseteq\cdot\cdot\cdot\subseteq \mathcal{P}(N,k_{r})$.
\end{lemma}

\begin{proof}
Firstly, consider the equal-error-probability rule based on the Bhattacharyya parameter.
By \cite[Lemma $4.7$]{polarchannelandsource}, if a BMS channel $V$ is a degraded version of $W$, then the subchannel $V_{N}^{(i)}$ is also degraded with respect to $W_{N}^{(i)}$ and $Z(V_{N}^{(i)})\geq Z(W_{N}^{(i)})$. Let the threshold be $2^{-N^{\beta}}$ for some $\beta<1/2$. The codewords are generated by $x^{1:N}=u^{\mathcal{I}}G_{\mathcal{I}}$, where $G_{\mathcal{I}}$ is the submatrix of $G$ whose rows are indexed by information set $\mathcal{I}$. The information sets for these two channels are respectively given by
\begin{eqnarray}
\left\{\begin{aligned}
&\mathcal{I}_{W}&=&\{i:Z(W_{N}^{(i)})< 2^{-N^{\beta}}\}, \\ \notag\
&\mathcal{I}_{V}&=&\{i:Z(V_{N}^{(i)})< 2^{-N^{\beta}}\}. \notag\
\end{aligned}\right.
\end{eqnarray}
Due to the fact that $Z(V_{N}^{(i)})\geq Z(W_{N}^{(i)})$, we have $\mathcal{I}_{V}\subseteq \mathcal{I}_{W}$. If we construct polar codes $\mathcal{P}(N,|\mathcal{I}_{W}|)$ over $W$ and $\mathcal{P}(N,|\mathcal{I}_{V}|)$ over $V$, $G_{\mathcal{I}_{V}}$ is a submatrix of $G_{\mathcal{I}_{W}}$. Therefore $\mathcal{P}(N,|\mathcal{I}_{V}|)\subseteq \mathcal{P}(N,|\mathcal{I}_{W}|)$.

From Lemma \ref{lemma:degraded}, the channel of the $\ell$-th level is always degraded with respect to the channel of the $(\ell+1)$-th level, and consequently, $\mathcal{P}(N,k_{\ell})\subseteq \mathcal{P}(N,k_{\ell+1})$, for $1\leq\ell<r$.

Then, consider the selection based on the error probability itself. The nesting relation still holds. This is because, by \cite[Lemma $3$]{Ido}, $P_e(V_{N}^{(i)})\geq P_e(W_{N}^{(i)})$ since $V_{N}^{(i)}$ is degraded with respect to $W_{N}^{(i)}$.
\end{proof}

\begin{remark}
Although it will not be used in this paper, it is worth pointing out that the nesting relation also holds if we select the channel indices according to a threshold on the mutual information. This is because, again by \cite[Lemma $3$]{Ido}, $I(V_{N}^{(i)})\leq I(W_{N}^{(i)})$ since $V_{N}^{(i)}$ is degraded with respect to $W_{N}^{(i)}$.
\end{remark}

However, the complexity of exact code construction for a BMS channel with a continuous output alphabet appears to
be exponential in the block length. A quantization method was proposed in \cite{Ido} which transforms a BMS channel with a continuous output alphabet to one with a finite output alphabet. Also, the authors of \cite{polarconstruction} proposed an approximation method to construct polar codes efficiently over any BMS
channel. We follow these methods to construct polar codes for the $\Lambda_{\ell-1}/\Lambda_{\ell}$ channel.
It was shown in \cite{polarconstruction, Ido} that with a sufficient number of quantization levels, the approximation error is negligible while the computational complexity is still $O(N\log N)$.

We illustrate the construction procedure with the example of the $\mathbb{Z}/2\mathbb{Z}$ channel. We need a collection of binary symmetric channels (BSCs) to approximate this $\mathbb{Z}/2\mathbb{Z}$ channel. The conditional PDF of the output after the mod-$2$ operation is given by
\begin{eqnarray}
f_{\sigma,2\mathbb{Z}}(y|x)=\frac{1}{\sqrt{2\pi}\sigma}\sum\limits_{j=-\infty}^{+\infty}\text{exp}\left(-\frac{(y-x+2
j)^2}{2\sigma^2}\right). \notag
\end{eqnarray}
\color{black} Note that the output $Y$ of $\mathbb{Z}/2\mathbb{Z}$ channel is in the Voronoi region $[-1,1)$ of $2\mathbb{Z}$. The channel is symmetric in the sense that $f_{\sigma,2\mathbb{Z}}(y|x=0)=f_{\sigma,2\mathbb{Z}}(\varpi(y)|x=1)$, where $\varpi$ is a permutation such that $\varpi(y)=(y+1) \mod 2\mathbb{Z}$ for any $y \in [-1,1)$.

Then, the output can be divided into several intervals $A_{i}$ and $\varpi (A_{i})$ for $1\leq i \leq K$, where $K$ denotes the quantization level, $A_{i} \subset [-0.5,0.5)$ and $\varpi(A_i) \subset [0.5,1) \cup [-1,-0.5)$. \color{black} The $i$-th BSC is chosen with probability $p_{i}$ and let the cross-over probability be $x_{i}$, which are given by
\begin{eqnarray}\label{eqn:chanQZ}
\left\{\begin{aligned}
p_{i}&=\int_{A_{i}}f_{\sigma,2\mathbb{Z}}(y|x=1)+f_{\sigma,2\mathbb{Z}}(y|x=0)dy,\\
x_{i}&=\frac{\int_{A_{i}}f_{\sigma,2\mathbb{Z}}(y|x=0)dy}{p_{i}}.
\end{aligned}\right.
\end{eqnarray}

The partition of the continuous alphabet is based on a function of the likelihood ratio \cite{Ido}
\begin{eqnarray}
\zeta_{y}=\frac{f_{\sigma,2\mathbb{Z}}(y|x=0)}{f_{\sigma,2\mathbb{Z}}(y|x=1)}. \notag\
\end{eqnarray}
Note that $\zeta_{y} \geq 1$ for $y \in [-0.5,0.5)$.

The symmetric capacity of $W$ is
\begin{eqnarray}
I(W)=\int_{0}^{1}(f(y|x=0)+f(y|x=1))C[\zeta_{y}]dy,
\label{eqn:symmetriccapaciy}
\end{eqnarray}
where $C[\zeta]$ for $\zeta\geq1$ is defined as
\begin{eqnarray}
C[\zeta]=1-\frac{\zeta}{\zeta+1}\text{log}\left(1+\frac{1}{\zeta}\right)-\frac{1}{\zeta+1}\text{log}(\zeta+1). \notag\
\end{eqnarray}

In our case, we let the maximum value of $C[\zeta]$ be $C_{\text{max}}=C[\zeta_{0}]$. For $1\leq i\leq K$, each interval is defined as
\begin{eqnarray}
A_{i}=\left\{y \in [-0.5,0.5):\frac{i-1}{K}C_{\text{max}}\leq C[\zeta_y] \leq \frac{i}{K}C_{\text{max}}\right\}. \notag\
\end{eqnarray}
Thus, the number of discrete output symbols is $2K$. \color{black}Notice that the above quantization process results in a degraded channel with respect to the original one \cite{Ido}. \color{black} According to \cite[Lemma 15]{Ido}, the difference in symmetric capacities of the discrete-output BMS channel and the original continuous-output channel can be bounded by $\frac{1}{K}$. \color{black}In numerical experiments, $K=64$ is sufficient to guarantee a capacity loss around $10^{-4}$ for a binary-input AWGN channel with capacity 0.5.\color{black}

With the discrete BMS channel, we use the merging algorithm in \cite{polarconstruction} to construct polar codes. The main idea is to perform the calculations approximately by restricting the number of output symbols in each level. Then the construction complexity is $O(N K^{2}\log K)$. The details are given in Algorithm \ref{alg:algo} and  Algorithm \ref{alg:deg}, where the function $b(x)=2\sqrt{x(1-x)}$ denotes the Bhattacharyya parameter\footnote{Using the binary entropy function $g(x)=-x \log (x)-(1-x)\log (1-x)$ would give an algorithm for the capacity rule.} of a BSC with cross-over probability $x$. Algorithm \ref{alg:algo} starts with the list $P_{\mathcal{Q}}=\{(p_{1},x_{1}),\cdots,(p_{K},x_{K})\}$, obtained from \eqref{eqn:chanQZ} by quantizing the channel transition PDF $f_{\sigma, 2\mathbb{Z}}(y|x)$. It generates a tree from a BMS channel $W$ as the root node according to the polarization rules \cite[eq. (19)]{polarcodes}. and \cite[eq. (20)]{polarcodes}. Suppose an intermediate BSC channel $\mathcal{W}$ from the polar transform is represented by $P_{\mathcal{W}}=\{(p_{1},x_{1}),\cdots,(p_{M},x_{M})\}$, where $M$ is its size. Then it applies to $\mathcal{W}$ the mass merging Algorithm \ref{alg:deg} on each level of the tree to reduce the size of the output alphabet for the next level. After each merging step of Algorithm \ref{alg:deg}, the size of ${\mathcal{W}}$ is decreased by 1. Finally, Algorithm \ref{alg:algo} returns upper bounds $\overline{P}_{e}(W_{N}^{(i)}, K)$ on the probability of error under SC decoding for the degraded bit channel $W_{N}^{(i)}$; the transmitting subchannels are chosen according to $\overline{P}_{e}(W_{N}^{(i)}, K)$.

We note that both the error probability $P_e(\mathcal{Q})$ and the Bhattacharyya parameter $Z(\mathcal{W})$ can be calculated from their lists of the BSC pairs instead of their channel transition probability mass functions. In fact, we have $P_e(\mathcal{Q})=\sum_{i=1}^K p_ix_i$ and $Z(\mathcal{W})=2\sum_{i=1}^M p_i\sqrt{x_i(1-x_i)}$.

\begin{algorithm}
\caption{Construction of Polar Codes}\label{alg:algo}
\textbf{Input}: BMS channel $W$ (the $\mathbb{Z}/2\mathbb{Z}$ channel) represented by channel transition PDF $f_{\sigma, 2\mathbb{Z}}(y|x)$, block length $N=2^m$, size of information set $k\leq N$, quantization level $K$\\
\textbf{Output}: A set of upper bounds on the error probabilities of $N$ subchannels and an index subset of $\{1,...,N\}$ of size $k$.
\begin{algorithmic}[1]\label{alg:1}
\State Calculate the Bhattacharyya parameter $\mathsf{Z}$ of $W$.

\State Quantize $W$ to $\mathcal{Q}$, represented by the root list $P_{\mathcal{Q}}=\{(p_{1},x_{1}),\cdots,(p_{K},x_{K})\}$, where $x_1 < \cdots < x_K$, using \eqref{eqn:chanQZ}.

\For{$i=1,2,...,N$}
    \State Express $i-1$ in binary representation $\langle b_1,b_2,...b_m\rangle$.
    \For{$j=1,2,...,m$}
       \If{$b_j=0$}
            \State Calculate the probability mass function of the worse polarized channel:
            \State $\mathcal{W} \gets \mathcal{Q} \boxtimes \mathcal{Q}$ \cite[eq. (19)]{polarcodes}. Obtain the BSC pairs $P_{\mathcal{W}}$ of $\mathcal{W}$.
            \State $\mathsf{Z} \gets \min\{Z(\mathcal{W}), 2\mathsf{Z}-\mathsf{Z}^2\}$.
       \Else
            \State Calculate the probability mass function of the better polarized channel:
            \State $\mathcal{W} \gets \mathcal{Q} \otimes \mathcal{Q}$ \cite[eq. (20)]{polarcodes}. Obtain the BSC pairs $P_{\mathcal{W}}$ of $\mathcal{W}$.
            \State $\mathsf{Z} \gets \mathsf{Z}^2$.
       \EndIf
    \State $\mathcal{Q} \gets$ degrading-merging($\mathcal{W},K$).
    \EndFor
\State Compute upper-bound $\overline{P}_{e}(W_{N}^{(i)}, K)=\min\{P_e(\mathcal{Q}),\mathsf{Z}\}$.
\EndFor
\State Return the set $\{\overline{P}_{e}(W_N^{(1)},K),...,\overline{P}_{e}(W_N^{(N)},K)\}$ and the subset of those $k$ indices with smallest $\{\overline{P}_{e}(W_N^{(i)},K)\}$.

\end{algorithmic}
\end{algorithm}

\begin{algorithm}
\caption{Degrading-merging function}\label{alg:deg}
\textbf{Input}: A list of BSC pairs $P_{\mathcal{W}}$, a quantization level $K$\\
\textbf{Output}: A list of BSC pairs $P_{\mathcal{Q}}$ with size $K$.
\begin{algorithmic}[1]
\While{$P_{\mathcal{W}}$ has size $>K$}
    \State Find the index $j=\underset{{i}}{\arg\min} \{p_i(b(\bar{x}_i)-b(x_i))-p_{i+1}(b(x_{i+1})-b(\bar{x}_i))\}$, where $\bar{x}_i=\frac{p_i x_i+p_{i+1} x_{i+1}}{p_i+p_{i+1}}$.
    \State Merge $(p_j, x_j)$ and $(p_{j+1},x_{j+1})$ into $(p_j+p_{j+1}, \bar{x}_j)$.
\EndWhile
\State Return $P_{\mathcal{W}}$.
\end{algorithmic}
\end{algorithm}

\color{black}Although the above merging algorithm results in an approximation error, it can be bounded properly by increasing the size of the finite output alphabet. \color{black} To this end, we introduce the \emph{capacity loss} $\epsilon_{\textrm{loss}}$ under the quantization-merging algorithm and finite length. More precisely, it means that we can construct a polar code of length $N$ over a channel with the symmetric capacity $C$ such that this polar code is assured to have a block error probability $P_{B}^{SC}\leq N2^{-N^{\beta}}$ ($\beta<1/2$) at the rate $C-\epsilon_{\textrm{loss}}$. We give the following lemma on the capacity loss, which is essentially an adaption of in \cite[Theorem 1]{Ido} to the $\mathbb{Z}/2\mathbb{Z}$ channel. This lemma shows that we can get arbitrarily close to the optimal construction of a polar code as $K$ increases.

\begin{lemma}
Given any constant $0<\beta<1/2$, define the capacity loss $\epsilon_{\textrm{loss}}$
\begin{eqnarray}
\frac{1}{N}\left|\{i:\overline{P}_{e}(W_{N}^{(i)}, K)<2^{-N^{\beta}}\}\right|= I(W)-\epsilon_{\textrm{loss}}. \notag\
\end{eqnarray}
For arbitrary real constant $\epsilon>0$, there exists a quantization level $K_0=K_0(W,\epsilon,\beta)$, determined by the underlying channel $W$, the constants $\epsilon$ and $\beta$, such that for all integers $K\geq K_0$ and all sufficiently large code lengths $N$, the polar code constructed from Algorithm \ref{alg:algo} within running time $O(N\cdot K^2\log K)$ achieves a rate loss $\epsilon_{\textrm{loss}} \leq \epsilon+\frac{1}{K}$ and a block error probability $P_{B}^{SC}\leq N2^{-N^{\beta}}$ under SC decoding.
\label{lem:rateloss}
\end{lemma}

\begin{proof}
Since \cite[Theorem 1]{Ido} addresses binary-input discrete symmetric channels, we need apply \cite[(57)]{Ido} to the quantized channel $Q$. However, it was only proved that for any $\epsilon>0$ and a sufficiently large $K \geq K_0(Q,\epsilon,\beta)$, the following $\lim\inf$ exists for the quantized channel $Q$.
\begin{eqnarray}
\underset{N \to \infty}{\lim\inf}\frac{1}{N}\left|\{i:\overline{P}_{e}(W_{N}^{(i)}, K)<2^{-N^{\beta}}\}\right| \geq I(Q)-\epsilon.
\end{eqnarray}
Note that $\overline{P}_{e}(W_{N}^{(i)}, K)$ denotes an upper bound on the error probability of the original subchannel $W_N^{(i)}$, which is returned by Algorithm \ref{alg:algo}. By a more recent work \cite[Lemma 1]{TalSimpleProof}, the above $\lim\inf$ can be safely replace by $\lim$ because the Bhatacharyya parameters of the subchannels from Algorithm \ref{alg:algo} eventually meet the form of \cite[Eq. (1)]{TalSimpleProof}. This can be checked from the two-staged polarization process introduce in the proof of \cite[Theorem 1]{Ido}.

\color{black}
As a result, for any $\epsilon>0$, there exists a sufficiently large $K \geq K_0(Q,\epsilon,\beta)$ and a sufficiently large block length $N$, such that the resulted polar code from Algorithm \ref{alg:algo} has a rate
\begin{eqnarray}\label{eqn:lemma4_1}
R=\frac{1}{N}\left|\{i:\overline{P}_{e}(W_{N}^{(i)}, K)<2^{-N^{\beta}}\}\right|\geq I(Q)-\epsilon,
\end{eqnarray}
and its block error probability under successive cancellation decoding satisfies $P_{B}^{SC}\leq N2^{-N^{\beta}}$.

Note that $Q$ is a quantized version of $W$ by the degradation merging process \eqref{eqn:chanQZ}. By \cite[Lemma 16]{Ido}, we have
\begin{eqnarray}\label{eqn:lemma4_2}
0 \leq I(W)-I(Q) \leq \frac{1}{K}.
\end{eqnarray}
Then, combining \eqref{eqn:lemma4_1} and \eqref{eqn:lemma4_2} gives us
\begin{eqnarray}
\epsilon_{\textrm{loss}} \leq \epsilon+\frac{1}{K},
\end{eqnarray}
which can be made arbitrarily small when $K$ is sufficiently large.
Since $Q$ is quantized from $W$, we may also write $K_0$ as $K_0(W,\epsilon,\beta)$. The proof is completed. \color{black}
\end{proof}

\begin{remark}
To remove the dependency of the rate loss in Lemma \ref{lem:rateloss} on the quantization level $K$, we may follow  \cite[Corollary 2]{Ido} to set $K=\lfloor\log N\rfloor$ in Lemma \ref{lem:rateloss}. Then, Algorithm \ref{alg:algo} will produce a polar code with rate loss $\epsilon_{\textrm{loss}} \leq \epsilon+\frac{1}{\lfloor\log N\rfloor}$ and block error probability $\leq N2^{-N^{\beta}}$, with complexity $O(N \log^2 N \log\log N)$. Clearly, $\epsilon_{\textrm{loss}} \to 0$ as $N \to \infty$. For a more detailed analysis on $\epsilon_{\textrm{loss}}$, see \cite{GuruswamiXia15}.
\end{remark}

\subsection{AWGN Goodness}

By combining the previous lemmas, we arrive at the main result of this section:

\begin{theorem}\label{theorem1}
Construct polar lattice $L$ with the $n$-dimensional binary lattice partition chain $\Lambda/\Lambda_1/\cdots/\Lambda_{r-1}/\Lambda'$ and $r$ nested polar codes of block length $N$, where $r= O(\log N)$ such that  $\epsilon_{\Lambda}(\sigma) = O(e^{-{N}})$ and $P_{e}(\Lambda',\sigma^{2}) = O(e^{-{N}})$. For any $0<\beta<1/2$, the error probability of $L$ under multistage decoding is bounded by
\begin{eqnarray}
P_{e}(L, \sigma^{2})\leq rN2^{-N^{\beta}}+N\left(1-\int_{\mathcal{V}(\Lambda')}f_{\sigma^{2}}(x)dx\right),
\label{eqn:errorbound}
\end{eqnarray}
with the logarithmic VNR bounded by \eqref{eqn:minimumVNR}. Then, $L$ is AWGN-good, i.e., $P_{e}(L, \sigma^{2})\to 0$ as $N\to \infty$ for arbitrary VNR greater than $2\pi e$.
\end{theorem}

\begin{proof}
The fact that the component polar codes are nested is due to Lemma \ref{lem:nested}, while the condition $r= O(\log N)$ is due to Lemma \ref{lem:numberoflevels}.
The error probability bound \eqref{eqn:errorbound} comes from \eqref{total-Pe}.
For a threshold $2^{-N^{\beta}}$ of the Bhattacharyya parameter, the block error probability of a polar code with SC decoding is upper-bounded by $N 2^{-N^{\beta}}$, which gives the first term on the right-hand side of \eqref{eqn:errorbound}. The second term of \eqref{eqn:errorbound} is due to the union bound. Since both terms of \eqref{eqn:errorbound} vanish as block length $N$ grows, $P_{e}(L, \sigma^{2})\to 0$.

Then we analyze the VNR. By Lemma \ref{lem:mod-capacity}, we have $\epsilon_{1}=C(\Lambda, \sigma^2) \leq \log(e) \cdot \epsilon_{\Lambda}(\sigma)=O(e^{-{N}})$.
Also, the capacity loss $\epsilon_{3}$ can be arbitrarily small as $N\to \infty$. Plugging these into \eqref{eqn:minimumVNR}, we can make
$\log\left(\frac{\gamma_{L}(\sigma)}{2\pi e}\right)$ arbitrarily close to $0$
as $N\to\infty$.
\end{proof}

\begin{remark}
In practice, if the target error probability is fixed (e.g., $10^{-5}$), $r$ can be a small integer, namely, $r$ does not have to scale as $\log N$. Thus, the essential condition is $N\to \infty$. Particularly, our example in Section IV shows that $r=2$ is sufficient for a target error probability around $10^{-5}$ when $\sigma=0.3380$.
\end{remark}

For finite $N$, however, the capacity loss $\epsilon_{3}$ is not negligible. We investigate the finite-length performance of polar lattices in the following.

The finite-length analysis of polar codes was given in \cite{Hassani2,ScalingXia,ScalingDina}. It was proved that for a fixed error probability, polar codes need a polynomial block length with respect to the gap to capacity $\epsilon_{\text{loss}}=I(W)-R=O(N^{-\frac{1}{\mu}})$ \cite{Hassani2,ScalingXia}, where $\mu$ is known as the scaling exponent. The lower bound of the gap is $\epsilon_{\text{loss}}\geq\underline{\beta}N^{-\frac{1}{\underline{\mu}}}$, where $\underline{\beta}$ is a constant that depends only on $I(W)$ and $\underline{\mu}=3.55$ \cite{Hassani2}. The upper bound of the gap is $\epsilon_{\text{loss}}\leq\bar{\beta}N^{-\frac{1}{\bar{\mu}}}$, where $\bar{\beta}$ is a constant that depends only on the block error probability $P_B$ and $\bar{\mu}=7$ was given in \cite{Hassani2}. Later this scaling factor $\bar{\mu}$ was improved to $5.77$ \cite{ScalingDina}.

Thus, the gap to the Poltyrev capacity of finite-dimensional polar lattices is
\begin{eqnarray}
\log\left(\frac{\gamma_{L}(\sigma)}{2\pi e}\right) \label{eqn:finitebound}
&\leq& \frac{2}{n}\left(\epsilon_1+r\bar{\beta}N^{-\frac{1}{\bar{\mu}}}\right) \notag\
\end{eqnarray}
with the corresponding block error probability
\begin{eqnarray}
P_{e}(L,\sigma^2)\leq rP_{B}+P_{e}(\Lambda'^{N},\sigma^{2}), \notag\
\end{eqnarray}
where the constant $\bar{\beta}$ depends only on $P_{B}$ (assuming equal error probabilities for the component polar codes). Since $n\ll N$ is fixed, the gap to the Poltyrev capacity of polar lattices also scales polynomially in the dimension $n_{L}=nN$.

In comparison, the optimal bound for finite-dimensional lattices is given by \cite{finiteIC}
\begin{eqnarray}
\log\left(\frac{\gamma_{L}(\sigma)}{2\pi e}\right)_{\text{opt}}=\sqrt{\frac{2}{n_{L}}}Q^{-1}(P_{e}(L,\sigma^2))-\frac{1}{n_{L}}\log n_{L}+ O\left(\frac{1}{n_{L}}\right).
\label{eqn:optimalfinitebound}
\end{eqnarray}
At finite dimensions, this is more precise than the exponential error bound for lattices constructed from random linear codes given in \cite{forney6}. Thus, given $P_{e}(L,\sigma^2)$, the scaling exponent of optimum random lattices is $2$ which is smaller than that of polar lattices $\bar{\mu}$. The result is consistent with the fact that polar codes require larger block length than random codes to achieve the same rate and error probability.

\section{Polarization-Based Gaussian Shaping}
\label{sec:shaping}

To achieve the capacity of the power-constrained Gaussian channel, we can apply Gaussian shaping over the polar lattice $L$. However, it appears difficult to do so directly. In this section, we will apply Gaussian shaping to the top lattice $\Lambda$ instead, which is more friendly for implementation. This is motivated by \cite[Theorem 2]{LingBel13}, which implies that one may construct a capacity-achieving lattice code from a good constellation. More precisely, one may choose a low-dimensional top lattice such as $\mathbb{Z}$ and $\mathbb{Z}^2$ whose mutual information has a negligible gap to the channel capacity as bounded in \cite[Theorem 2]{LingBel13}, and then construct a multilevel code to achieve the capacity. We will show that this strategy is equivalent to implementing Gaussian shaping over the AWGN-good polar lattice.

\subsection{Asymmetric Channels in Multilevel Lattice Coding}
\label{seq:goodconstellation}
By \cite[Theorem 2]{LingBel13}, we choose a constellation $D_{\Lambda,\sigma_s}$ such that the flatness factor $\epsilon_{\Lambda}\left(\tilde{\sigma}\right)$ is negligible, where $\tilde{\sigma}=\frac{\sigma_s\sigma}{\sqrt{\sigma_s^2+\sigma^2}}$. Let the binary partition chain $\Lambda/\Lambda_{1}/\cdots/\Lambda_{r-1}/\Lambda'/\cdots$ be labelled by bits $X_{1},\cdots,X_{r},\cdots$. Then, $D_{\Lambda,\sigma_s}$ induces a distribution $P_{X_{1:r}}$ whose limit corresponds to $D_{\Lambda,\sigma_s}$ as ${r\rightarrow\infty}$. An example for $D_{\mathbb{Z},\sigma_s}$ for $\sigma_s=3$ is shown in Figure \ref{fig:LatticeGaussian}. In this case, a shaping constellation with $M=32$ ($r=5$) points are actually sufficient, since the total probability of these points is rather close to~$1$.

\begin{figure*}[ht!]
    \begin{center}
        \subfigure[$D_{\mathbb{Z},\sigma_s}$ for $\sigma_s=3$]{%
            \label{}
            \includegraphics[width=0.45\textwidth]{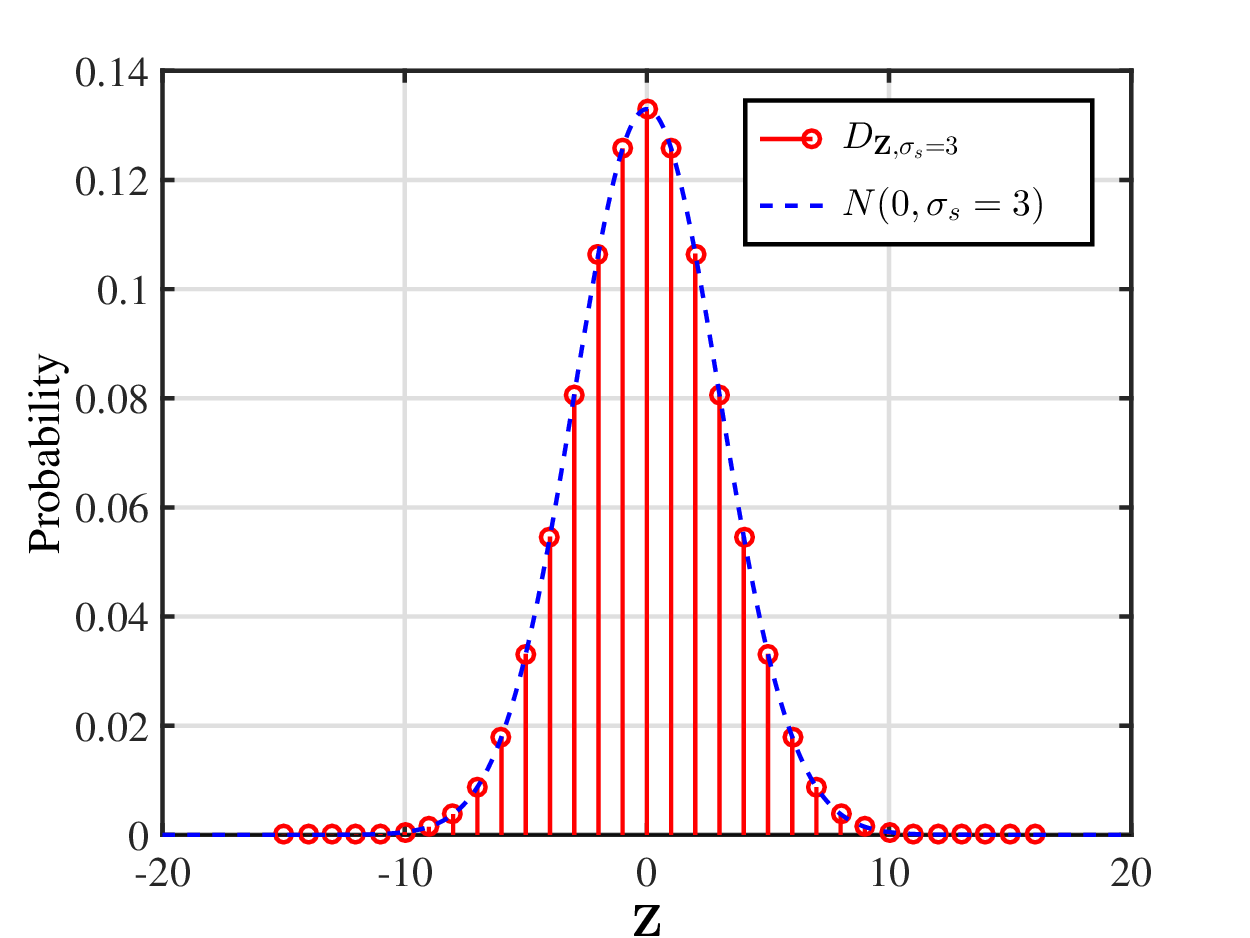}
        }%
        \subfigure[Bit labelling]{%
           \label{}
           \includegraphics[width=0.45\textwidth]{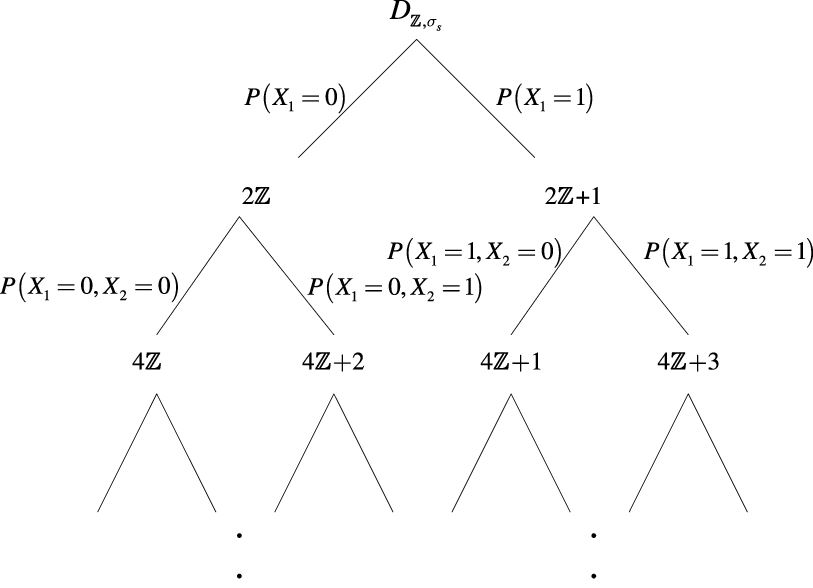}
        }\\ 
%
    \end{center}

    \vspace{-12pt}

    \caption{%
           Lattice Gaussian distribution $D_{\mathbb{Z},\sigma_s}$ and the associated labelling. A probability $P(X_1, X_2, ..., X_i)$ in (b) is given by that of the coset indexed by bits $X_1, X_2, ..., X_i$; for example, $P(X_1=1, X_2=0)=\sum_{\lambda \in 4\mathbb{Z}+1} \text{Pr}(\lambda)$, where $\text{Pr}(\cdot)$ denotes the probability mass function of $D_{\mathbb{Z},\sigma_s}$.
     }%
    \label{fig:LatticeGaussian}
\end{figure*}


By the chain rule of mutual information
\begin{eqnarray}\label{eqn:chainrule}
I(Y;X_{1:r})=\sum_{\ell=1}^{r} I(Y;X_\ell|X_{1:\ell-1}),
\end{eqnarray}
we obtain $r$ binary-input channels ${W}_{\ell}$ for $1\leq \ell \leq r$. Given $x_{1:\ell-1}$, denote again by $\mathcal{A}_\ell(x_{1:\ell})$ the coset of $\Lambda_{\ell}$ indexed by $x_{1:\ell-1}$ and $x_{\ell}$. According to \cite{multilevel}, the channel transition PDF of the $\ell$-th channel ${W}_{\ell}$ is given by
\vspace{-1em}
\begin{eqnarray}
\label{eqn:transition}
& & \hspace{-3em}P_{Y|X_\ell,X_{1:\ell-1}}(y|x_\ell,x_{1:\ell-1}) \notag\\
&=& \frac{1}{P\{\mathcal{A}_\ell(x_{1:\ell})\}}\sum_{a\in \mathcal{A}_\ell(x_{1:\ell})}P(a)P_{Y|A}(y|a)  \notag\\
&=& \frac{1}{f_{\sigma_s}(\mathcal{A}_\ell(x_{1:\ell}))}\sum_{a\in \mathcal{A}_\ell(x_{1:\ell})}\frac{1}{2\pi \sigma\sigma_s}\text{exp}\left(-\frac{\|y-a\|^2}{2\sigma^{2}}-\frac{\|a\|^2}{2\sigma_s^{2}}\right) \\
&=&\text{exp}\left(-\frac{\|y\|^2}{2(\sigma^2_s+\sigma^2)}\right)\frac{1}{f_{\sigma_s}(\mathcal{A}_\ell(x_{1:\ell}))}\frac{1}{2\pi \sigma\sigma_s}\sum_{a\in \mathcal{A}_\ell(x_{1:\ell})}\text{exp}\left(-\frac{\sigma_s^2+\sigma^2}{2\sigma_s^2\sigma^2}\left\|\frac{\sigma_s^2}{\sigma_s^2+\sigma^2}y-a\right\|^2\right) \notag\\
&=&\text{exp}\left(-\frac{\|y\|^2}{2(\sigma_s^2+\sigma^2)}\right)\frac{1}{f_{\sigma_s}(\mathcal{A}_\ell(x_{1:\ell}))}\frac{1}{2\pi\sigma\sigma_s}\sum_{a\in \mathcal{A}_\ell(x_{1:\ell})}\text{exp}\left(-\frac{\|\alpha y-a\|^2}{2\tilde{\sigma}^2}\right). \notag
\end{eqnarray}
where $\alpha=\frac{\sigma_s^2}{\sigma_s^2+\sigma^2}$ is the MMSE coefficient. In general, ${W}_{\ell}$ is asymmetric with the input distribution $P_{X_{\ell}|X_{1:\ell-1}}$ unless $f_{\sigma_s}(\mathcal{A}_\ell(x_{1:\ell}))/f_{\sigma_s}(\mathcal{A}_{\ell-1}(x_{1:\ell-1}))\approx\frac{1}{2}$, which means that $\epsilon_{\Lambda_{\ell}}(\sigma_s)$ is negligible.

For a finite power, the number of levels does not need to be large. The following
lemma shows in a quantitative manner how large $r$ should be in order to achieve the channel capacity. The proof can be found in Appendix \ref{AppendixIloss}.
\begin{lemma}\label{lem:Iloss}
There exists $r= O(\log\log N)$ such that  using the first $r$ levels only incurs a capacity loss~$\sum_{\ell > r} I(Y;X_{\ell}|X_{1:\ell-1}) = O(\frac{1}{N})$.
\end{lemma}

\begin{remark}
The condition $r= O(\log\log N)$ is of theoretical interest, similarly to the condition $r = O(\log N)$ in the AWGN-good setting (Lemma \ref{lem:numberoflevels}). In practice, $r$ can be a small constant so that the different between $I(Y;X_{1:r})$ and capacity is negligible, as we will see from the example in the next section. Note that the relaxed condition on $r$ is thankfully due to the power constraint. Unlike the AWGN-good setting, here we no longer have the term $P_e({\Lambda'}^N,\sigma^2)$ in the upper bound of the error probability, since the bottom lattice $\Lambda'$ does not carry any message.
\end{remark}

\subsection{Polar Codes for Asymmetric Channels}

Since the component channels are asymmetric, we need polar codes for asymmetric channels to achieve their capacity. Fortunately, polar codes for the binary memoryless asymmetric (BMA) channels have been introduced in \cite{aspolarcodes,Mondelli18} recently.
\label{sec:polarAsym}
\begin{definition}[Bhattacharyya Parameter for BMA Channel \cite{polarsource,aspolarcodes}]\label{deft:asymZ}
Let $W$ be a BMA channel with input $X \in \mathcal{X}=\{0,1\}$ and output $Y \in \mathcal{Y}$, and let $P_X$ and $P_{Y|X}$ denote the input distribution and channel transition probability, respectively. The Bhattacharyya parameter $Z$ for channel $W$ is the defined as
\begin{eqnarray}
Z(X|Y)&=&2\sum\limits_{y}P_Y(y)\sqrt{P_{X|Y}(0|y)P_{X|Y}(1|y)} \notag\ \\
&=&2\sum\limits_{y}\sqrt{P_{X,Y}(0,y)P_{X,Y}(1,y)}. \notag
\end{eqnarray}
\end{definition}
Note that this definition reduces to that for the BMS channel when $P_X$ is uniform.

The next lemma shows that adding an observable at the output of $W$ will not increase $Z$.
\begin{lemma}[Conditioning reduces Bhattacharyya parameter $Z$]\label{lem:CondiReduce}
Let $(X,Y,Y')\sim P_{X,Y,Y'}, X\in\mathcal{X}=\{0,1\}, Y\in \mathcal{Y},Y'\in \mathcal{Y}'$, we have
\begin{eqnarray}
Z(X|Y,Y')\leq Z(X|Y). \notag
\end{eqnarray}
\end{lemma}

\begin{proof}
\begin{eqnarray}
Z(X|Y,Y')&=&2\sum\limits_{y,y'}\sqrt{P_{X,Y,Y'}(0,y,y')P_{X,Y,Y'}(1,y,y')} \notag \\
&=&2\sum\limits_{y}\sum\limits_{y'}\sqrt{P_{X,Y,Y'}(0,y,y')}\sqrt{P_{X,Y,Y'}(1,y,y')} \notag \\
&\stackrel{(a)}\leq& 2\sum\limits_{y}\sqrt{\sum\limits_{y'}P_{X,Y,Y'}(0,y,y')}\sqrt{\sum\limits_{y'}P_{X,Y,Y'}(1,y,y')} \notag \\
&=&2\sum\limits_{y}\sqrt{P_{X,Y}(0,y)P_{X,Y}(0,y)} \notag
\end{eqnarray}
where $(a)$ follows from Cauchy-Schwarz inequality.
\end{proof}

Let $X^{1:N}$ and $Y^{1:N}$ be the input and output vector after $N$ independent uses of $W$. For simplicity, denote the distribution of $(X^i,Y^i)$ by $P_{XY}=P_X P_{Y|X}$ for $i \in [N]$. The following property of the polarized random variables $U^{1:N}=X^{1:N}G_N$ is well known.

\begin{theorem}[Polarization of Random Variables \cite{aspolarcodes}]\label{theorem:polarization}
For any $\beta\in(0,1/2)$,
\begin{eqnarray}
\begin{aligned}
&\left.\left.\lim_{N\rightarrow\infty}\frac{1}{N}\right\vert\left\{i:Z(U^i|U^{1:i-1})\geq1-2^{-N^{\beta}}\right\}\right\vert=H(X),\\
&\left.\left.\lim_{N\rightarrow\infty}\frac{1}{N}\right\vert\left\{i:Z(U^i|U^{1:i-1})\leq2^{-N^{\beta}}\right\}\right\vert=1-H(X),\\
&\left.\left.\lim_{N\rightarrow\infty}\frac{1}{N}\right\vert\left\{i:Z(U^i|U^{1:i-1},Y^{1:N})\geq1-2^{-N^{\beta}}\right\}\right\vert=H(X|Y),\\
&\left.\left.\lim_{N\rightarrow\infty}\frac{1}{N}\right\vert\left\{i:Z(U^i|U^{1:i-1},Y^{1:N})\leq2^{-N^{\beta}}\right\}\right\vert=1-H(X|Y),
\end{aligned}
\end{eqnarray}
and
\begin{eqnarray}
\begin{aligned}
&\left.\left.\lim_{N\rightarrow\infty}\frac{1}{N}\right\vert\left\{i:Z(U^i|U^{1:i-1},Y^{1:N})\leq2^{-N^{\beta}}\text{ and }Z(U^i|U^{1:i-1})\geq1-2^{-N^{\beta}}\right\}\right\vert=I(X;Y),\\
&\left.\left.\lim_{N\rightarrow\infty}\frac{1}{N}\right\vert\left\{i:Z(U^i|U^{1:i-1},Y^{1:N})\geq2^{-N^{\beta}}\text{ or }Z(U^i|U^{1:i-1})\leq1-2^{-N^{\beta}}\right\}\right\vert=1-I(X;Y).
\end{aligned}
\end{eqnarray}
\end{theorem}

\medskip

The Bhattacharyya parameter for asymmetric models was originally defined for distributed source coding in \cite{polarsource}. By the duality between channel coding and source coding, it can be also used to construct capacity-achieving polar codes for BMA channels \cite{aspolarcodes}. Actually, $Z(U^i|U^{1:i-1})$ is the Bhattacharyya parameter for a single source $X$ (without side information).

The Bhattacharyya parameter of a BMA channel can be related to that of a symmetric channel. To this aim, we use a symmetrization technique which creates a BMS channel $\tilde{W}$ from the BMA channel $W$ \cite{polarchannelandsource,aspolarcodes}.

\begin{lemma}[Symmetrization]\label{lem:Asy2sym}
Let $\tilde{W}$ be a binary-input channel with input $\tilde{X} \in \mathcal{X}=\{0,1\}$ and output $\tilde{Y} \in \mathcal{Y}\times\mathcal{X}$, built from the asymmetric channel $W$ by treating $\tilde{X}$ as the new input and $\tilde{X}\oplus X$ as an additional output, as shown in Figure \ref{fig:2Asym2sym}. Then $\tilde{W}$ is a binary-input symmetric channel in the sense that $P_{\tilde{Y}|\tilde{X}}(y,x\oplus\tilde{x}|\tilde{x})=P_{Y,X} (y,x)$. Therefore, the optimal input distribution of $\tilde{W}$ is the uniform distribution.
\end{lemma}

\begin{figure}[h]
    \centering
    \includegraphics[width=7cm]{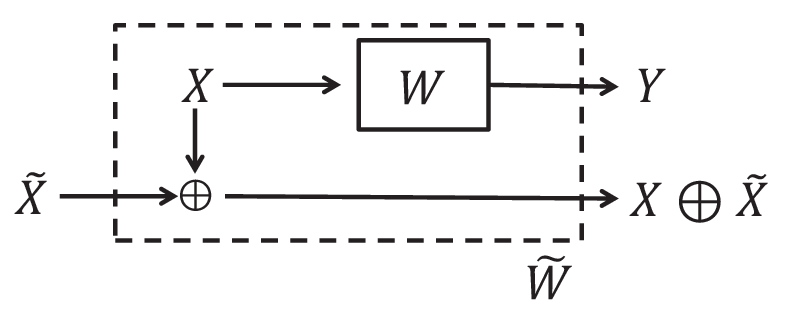}
    \caption{The relationship between the asymmetric channel $W$ and the symmetrized channel $\tilde{W}$.}
    \label{fig:2Asym2sym}
\end{figure}

\begin{proof}
\begin{eqnarray}
P_{\tilde{Y}|\tilde{X}}(y,x\oplus\tilde{x}|\tilde{x})&=& \frac{P_{\tilde{Y},\tilde{X}}(y,x\oplus\tilde{x},\tilde{x})}{P_{\tilde{X}}(\tilde{x})}
= \frac{\sum_{x'\in \mathcal{X}} P_{\tilde{Y},X,\tilde{X}}(y,x\oplus\tilde{x},x',\tilde{x})}{P_{\tilde{X}}(\tilde{x})} \notag \\
&\stackrel{(a)}=& \frac{\sum_{x'\in X} P_{Y|X}(y|x') P_{X\oplus\tilde{X},X,\tilde{X}}(x\oplus\tilde{x},x',\tilde{x})}{P_{\tilde{X}}(\tilde{x})} \notag \\
&\stackrel{(b)}=& \frac{\sum_{x'\in X} P_{Y|X}(y|x') P_{X\oplus\tilde{X}|X,\tilde{X}}(x\oplus\tilde{x}|x',\tilde{x}) P_{X}(x')P_{\tilde{X}}(\tilde{x})}{P_{\tilde{X}}(\tilde{x})} \notag \\
&\stackrel{(c)}=& P_{Y,X} (y,x). \notag
\end{eqnarray}
The equalities $(a)$-$(c)$ follow from $(a)$ $Y$ is only dependent of $X$, $(b)$ $X$ and $\tilde{X}$ are independent of each other and $(c)$ $P_{X\oplus\tilde{X}|X,\tilde{X}}(x\oplus\tilde{x}|x',\tilde{x})=\mathds{1}(x'=x)$.
\end{proof}

The following theorem connects the Bhattacharyya Parameter of a BMA channel $W$ and that of the symmetrized channel $\tilde{W}$. Denote by $W_N$ and $\tilde{W}_N$ the combining channels of $N$ uses of $W$ and $\tilde{W}$, respectively.

\begin{theorem}[Connection Between Bhattacharyya Parameters \cite{aspolarcodes}]\label{theorem:conasypolar}
Let $\tilde{X}^{1:N}$ and $\tilde{Y}^{1:N}=\left(Y^{1:N}, X^{1:N}\oplus\tilde{X}^{1:N}\right)$ be the input and output vectors of $\tilde{W}$, respectively, and let  $U^{1:N}$=$X^{1:N}G_N$ and $\tilde{U}^{1:N}$=$\tilde{X}^{1:N}G_N$. The Bhattacharyya parameter of each subchannel of $W_N$ is equal to that of each subchannel of $\tilde{W}_N$, i.e.,
\begin{eqnarray}
Z(U^i|U^{1:i-1},Y^{1:N})=\tilde{Z}(\tilde{W}_N^{(i)})=Z(\tilde{U}^i|\tilde{U}^{1:i-1},Y^{1:N}, X^{1:N}\oplus\tilde{X}^{1:N}). \notag\
\label{eqn:asymmetricz}
\end{eqnarray}
\end{theorem}

Now, we are in a position to construct polar codes for the BMA channel. Define the frozen set $\mathcal{\tilde{F}}$ and information set $\mathcal{\tilde{I}}$ of the symmetric polar codes as follows:
\begin{eqnarray}
\begin{cases}
\begin{aligned}
&\text{frozen set: } \mathcal{\tilde{F}}=\{i\in[N]:Z(U^i|U^{1:i-1},Y^{1:N})>2^{-N^{\beta}}\}\\
&\text{information set: } \mathcal{\tilde{I}}=\{i\in[N]:Z(U^i|U^{1:i-1},Y^{1:N})\leq2^{-N^{\beta}}\}.
\end{aligned}
\end{cases}
\end{eqnarray}

By Theorem \ref{theorem:conasypolar}, the Bhattacharyya parameters of the symmetrized channel $\tilde{W}$ and the asymmetric channel $W$ are the same. However, the channel capacity of $\tilde{W}$ is $I(\tilde{X};X\oplus\tilde{X})+I(\tilde{X};Y|X\oplus\tilde{X})=1-H(X)+I(X;Y)$,
which is $1-H(X)$ more than the capacity of $W$. To obtain the real capacity $I(X;Y)$ of $W$, the input distribution of $W$ needs to be adjusted to $P_X$. By polar lossless source coding, the indices with very small $Z(U^i|U^{1:i-1})$ should be removed from the information set $\mathcal{\tilde{I}}$ of the symmetrized channel, and the proportion of this part is $1-H(X)$ as $N\to \infty$. We name the remaining set as the information set $\mathcal{I}$ of the asymmetric channel $W$. Further, there are some bits which are uniformly distributed and can be made independent from the information bits; we name this set as the frozen set $\mathcal{F}$. In order to generate the desired input distribution $P_X$, the remaining bits are determined by the bits in $\mathcal{F}\cup\mathcal{I}$; we call it the shaping set $\mathcal{S}$. This process is depicted in Figure \ref{fig:symandasym}. We formally define the three sets as follows:
\begin{eqnarray}\label{eqn:asymdefinition}
\begin{cases}
\begin{aligned}
&\text{frozen set: } \mathcal{F}=\{i\in[N]:Z(U^i|U^{1:i-1},Y^{1:N})\geq1-2^{-N^{\beta}}\}\\
&\text{information set: } \mathcal{I}=\{i\in[N]:Z(U^i|U^{1:i-1},Y^{1:N})\leq2^{-N^{\beta}}\text{ and }Z(U^i|U^{1:i-1})\geq1-2^{-N^{\beta}}\}\\
&\text{shaping set: } \mathcal{S}=\left(\mathcal{F}\cup\mathcal{I}\right)^c.
\end{aligned}
\end{cases}
\end{eqnarray}

\begin{figure}[t]
    \centering
    \label{fig:1stPolar}
    \includegraphics[width=11cm]{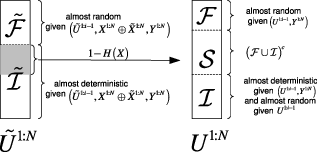}
    \caption{Polarization for symmetric and asymmetric channels.}
    \label{fig:symandasym}
\end{figure}

To find these sets, one can use Theorem \ref{theorem:conasypolar} to calculate $Z(U^i|U^{1:i-1},Y^{1:N})$ with the known technique for symmetric polar codes \cite{Ido}. We note that $Z(U^i|U^{1:i-1})$ can be computed in a similar way: one constructs a symmetric channel between $\tilde{X}$ and $X \oplus \tilde{X}$, which is actually a binary-input symmetric channel with cross-over probability $P_X(x=1)$. The above construction is equivalent to implementing shaping over the polar code for the symmetrized channel $\tilde{W}$.

Besides the construction, the decoding can also be converted to that of the symmetric polar code. If $X^{1:N}\oplus \tilde{X}^{1:N}=0$, we have $U^{1:N}=\tilde{U}^{1:N}$, which means the decoding result of $U^{1:N}$ equals to that of $\tilde{U}^{1:N}$. Thus, decoding of the polar code for $W$ can be treated as decoding of the polar code for $\tilde{W}$ given that $X\oplus \tilde{X}=0$. Clearly, the SC decoding complexity for asymmetric channel is also $O(N\log N)$. We summarize this observation as the following lemma.
\begin{lemma}[Decoding for Asymmetric Channel \cite{aspolarcodes}]\label{lem:AsymDec}
Let $y^{1:N}$ be a realization of $Y^{1:N}$ and $\hat{u}^{1:i-1}$ be the previous $i-1$ estimates of $u^{1:N}$. The likelihood ratio of $u^i$ is given by
\begin{eqnarray}\label{eqn:asydec}
\frac{P_{U^i|U^{1:i-1},Y^{1:N}}(0|\hat{u}^{1:i-1},y^{1:N})}{P_{U^i|U^{1:i-1},Y^{1:N}}(1|\hat{u}^{1:i-1},y^{1:N})}=
\frac{\tilde{W}_N^{(i)}((y^{1:N},0^{1:N}),\hat{u}^{1:i-1}|0)}{\tilde{W}_N^{(i)}((y^{1:N},0^{1:N}),\hat{u}^{1:i-1}|1)},
\end{eqnarray}
where $\tilde{W}_N^{(i)}$ denotes the transition probability of the $i$-th subchannel of $\tilde{W}_N$.
\end{lemma}

In \cite{aspolarcodes}, the bits in $\mathcal{F} \cup \mathcal{S}$ are all chosen according to $P_{U^i|U^{1:i-1}}(u^i|u^{1:i-1})$, which can also be calculated using \eqref{eqn:asydec} (treating $Y$ as an independent variable and remove it). However, in order to be compatible with polar lattices, we modify the scheme such that the bits in $\mathcal{F}$ are uniformly distributed over $\{0,1\}$ while the bits in $\mathcal{S}$ are still chosen according to $P_{U^i|U^{1:i-1}}(u^i|u^{1:i-1})$. The expectation of the decoding error probability still vanishes with $N$. The following theorem is an extension of the result in \cite[Theorem 3]{aspolarcodes}, and its proof is given in Appendix \ref{appendix0}.

\begin{theorem}\label{theorem:codingtheorem}
Consider a polar code with the following encoding and decoding strategies for a BMA channel.
\begin{itemize}
\item Encoding: Before sending the codeword $x^{1:N}=u^{1:N}G_N$, the index set $[N]$ are divided into three parts: the frozen set $\mathcal{F}$, the information set $\mathcal{I}$ and the shaping set $\mathcal{S}$ which are defined in \eqref{eqn:asymdefinition}. The encoder places uniformly distributed information bits in $\mathcal{I}$, and fills $\mathcal{F}$ with a uniform random $\{0,1\}$ sequence which is shared between the encoder and the decoder. The bits in $\mathcal{S}$ are generated by a mapping $\phi_{\mathcal{S}}\triangleq\{\phi_i\}_{i\in\mathcal{S}}$ in the family of randomized mappings $\Phi_{\mathcal{S}}$, which yields the following distribution:
\begin{equation}
u^i=\phi_i(u^{1:i-1})=
\begin{cases}
0 \;\;\;\; \text{with probability } P_{U^i|U^{1:i-1}}(0|u^{1:i-1}),\notag\ \\
1 \;\;\;\; \text{with probability } P_{U^i|U^{1:i-1}}(1|u^{1:i-1}). \notag\
\end{cases}
\label{eqn:encoding}
\end{equation}
\item Decoding: The decoder receives $y^{1:N}$ and estimates $\hat{u}^{1:N}$ of $u^{1:N}$ according to the rule
\begin{equation}
\hat{u}^i=
\begin{cases}
u^i, \;\;\;\;\;\;\;\;\;\;\;\;\;\;\;\;\text{if } i\in\mathcal{F} \\
\phi_i(\hat{u}^{1:i-1}),   \;\;\;\;\; \text{if } i\in\mathcal{S} \\
\underset{u}{\operatorname{argmax}} \; P_{U^i|U^{1:i-1},Y^{1:N}}(u|\hat{u}^{1:i-1},y^{1:N}), \,\text{if } i\in\mathcal{I}. \notag\
\end{cases}.
\label{eqn:decoding1}
\end{equation}
\end{itemize}
With the above encoding and decoding, the message rate can be arbitrarily close to $I(X;Y)$ and the expectation of the decoding error probability over the randomized mappings satisfies $E_{\Phi_{\mathcal{S}}}[P_e^{SC}(\phi_{\mathcal{S}})]= O(2^{-N^{\beta'}})$ for $\beta'<\beta<0.5$, where $\beta$ is used to choose the frozen set, the information set, and the shaping set as in \eqref{eqn:asymdefinition}.
\end{theorem}

\medskip

By an averaging argument, there exists a deterministic mapping $\phi_{\mathcal{S}}$ such that $P_e(\phi_{\mathcal{S}})= O(2^{-N^{\beta'}})$. However, it is difficult to actually find such a deterministic mapping. In practice, we may share a random mapping $\phi_{\mathcal{S}}$ between the encoder and decoder, i.e., let them have access to the same source of randomness (e.g., using the same seed for the pseudorandom number generators).

\subsection{Multilevel Polar Codes}

Next, our task is to construct polar codes to achieve the mutual information $I(Y;X_\ell|X_{1:\ell-1})$ for all levels. The construction of the preceding subsection is readily applicable to the construction for the first level $W_1$. To demonstrate the construction for other levels, we take the channel of the second level $W_2$ as an example. This is also a BMA channel with input $X_2\in\mathcal{X}=\{0,1\}$, output $Y\in\mathcal{Y}$ and side information $X_1$. Its channel transition probability is shown in \eqref{eqn:transition}. To construct a polar code for the second level, we propose the following two-step procedure.

\begin{figure}[t]
    \centering
    \includegraphics[width=13cm]{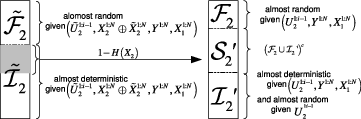}
    \caption{The first step of polarization in the construction for the second level.}
    \label{fig:2ndsymandasym}
\end{figure}

\begin{enumerate}[{Step}~1:]
\item Construct a polar code for the BMS channel with input vector $\tilde{X}_2^{1:N}=[\tilde{X}_2^1,\tilde{X}_2^2,\cdot\cdot\cdot,\tilde{X}_2^N]$ and output vector $\tilde{Y}^{1:N}=\left(X_2^{1:N}\oplus\tilde{X}_2^{1:N},Y^{1:N},X_1^{1:N}\right)$ where $\tilde{X}_2^i\in\mathcal{X}=\{0,1\}$ is uniformly distributed. At this step $X_1$ is regarded as a part of the outputs. Then the distribution of $X_2$ becomes the marginal distribution $\sum_{x_1,x_{3:r}}P_{X_{1:r}}(x_{1:r})$. Consider polarized random variables ${U}_2^{1:N}={X}_2^{1:N}G_N$ and $\tilde{U}_2^{1:N}=\tilde{X}_2^{1:N}G_N$. According to Theorem \ref{theorem:polarization}, the polarization gives us the three sets $\mathcal{F}_2$, $\mathcal{I}'_2$ and $\mathcal{S}'_2$ as shown in Figure \ref{fig:2ndsymandasym}. Similarly, we can prove that $\frac{|\mathcal{I}'_2|}{N}\rightarrow I(Y,X_1;X_2)$ and $\frac{|\mathcal{F}_2 \cup \mathcal{S}'_2|}{N}\rightarrow 1-I(Y,X_1;X_2)$ as $N\to \infty$. These three sets are defined as follows:
    \begin{eqnarray}\label{eqn:asymdefinition2nd}
\begin{cases}
\begin{aligned}
&\text{frozen set: } \mathcal{F}_2=\{i\in[N]:Z(U_2^i|U_2^{1:i-1},Y^{1:N},X_1^{1:N})\geq1-2^{-N^{\beta}}\}\\
&\text{information set: } \mathcal{I}'_2=\{i\in[N]:Z(U_2^i|U_2^{1:i-1},Y^{1:N},X_1^{1:N})\leq2^{-N^{\beta}}\text{and } \\
&\hspace{9em}Z(U_2^i|U_2^{1:i-1})\geq1-2^{-N^{\beta}}\}\\
&\text{shaping set: } \mathcal{S}'_2=\left(\mathcal{F}_2\cup\mathcal{I}'_2\right)^c.
\end{aligned}
\end{cases}
\end{eqnarray}

\item Treat $X_1^{1:N}$ as the side information for the encoder. Given $X_1^{1:N}$, the choices of ${X}_2^{1:N}$ are further restricted since $X_1$ and $X_2$ are generally correlated, i.e., $P_{X_1,X_2}(x_1,x_2)=f_{\sigma_s}(\mathcal{A}(x_1,x_2))/f_{\sigma_s}(\Lambda)$ (cf. Figure \ref{fig:LatticeGaussian}). By removing from $\mathcal{I}'_2$ the bits which are almost deterministic given $U_2^{1:i-1}$ and $X_1^{1:N}$, we obtain the information set $\mathcal{I}_2$ for $W_2$. Then the distribution of the input $X_2$ becomes the conditional distribution $P_{X_2|X_1}(x_2|x_1)$. The process is shown in Figure \ref{fig:2nd}.
    More precisely, the indices are divided into three portions as follows:
\begin{eqnarray}
   \notag\ 1&=&\underbrace{1-I(\tilde{X}_2;\tilde{X}_2\oplus X_2,X_1,Y)}_{|\mathcal{F}_2|/N}+I(\tilde{X}_2;\tilde{X}_2\oplus X_2,X_1,Y) \\ \notag\
    &\stackrel{\text{Step} 1}=&\underbrace{1-I(\tilde{X}_2;\tilde{X}_2\oplus X_2,X_1,Y)}_{|\mathcal{F}_2|/N}+\underbrace{I(\tilde{X}_2;\tilde{X}_2\oplus X_2)}_{|\mathcal{S}'_2|/N}+\underbrace{I(\tilde{X}_2;X_1,Y|\tilde{X}_2\oplus X_2)}_{|\mathcal{I}'_2|/N} \\ \notag\
    &\stackrel{\text{Step} 2}=& \underbrace{1-I(\tilde{X}_2;\tilde{X}_2\oplus X_2,X_1,Y)}_{|\mathcal{F}_2|/N}+\underbrace{1-H(X_2)}_{|\mathcal{S}'_2|/N}+\underbrace{I(X_2;X_1)}_{|\mathcal{S}_{X_1}|/N}+\underbrace{I(X_2;Y|X_1)}_{|\mathcal{I}_2|/N} \\\notag\
    &=&\underbrace{1-I(\tilde{X}_2;\tilde{X}_2\oplus X_2,X_1,Y)}_{|\mathcal{F}_2|/N}+\underbrace{1-H(X_2|X_1)}_{|\mathcal{S}_2|/N}+\underbrace{I(X_2;Y|X_1)}_{|\mathcal{I}_2|/N}
\end{eqnarray}

We give the formal statement of this procedure in the following lemma.
\end{enumerate}
\begin{figure}[t]
    \centering
    \includegraphics[width=11cm]{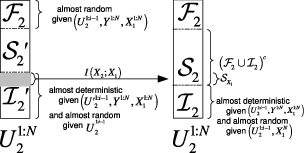}
    \caption{The second step of polarization in the construction for the second level.}
    \label{fig:2nd}
\end{figure}
\begin{lemma}
After the first step of polarization, we obtain the three sets $\mathcal{F}_2$, $\mathcal{I}'_2$ and $\mathcal{S}'_2$ in \eqref{eqn:asymdefinition2nd}. Let $\mathcal{S}_{X_1}$ denote the set of indices whose Bhattacharyya parameters satisfy $Z(U_2^i|U_2^{1:i-1},Y^{1:N},X_1^{1:N})\leq2^{-N^{\beta}}$, $Z(U_2^i|U_2^{1:i-1},X_1^{1:N})\leq1-2^{-N^{\beta}}$ and $Z(U_2^i|U_2^{1:i-1})\geq1-2^{-N^{\beta}}$. The proportion of $\mathcal{S}_{X_1}$ is asymptotically given by $\lim_{N\rightarrow\infty}\frac{|\mathcal{S}_{X_1}|}{N}=I(X_2;X_1)$. Then by removing $\mathcal{S}_{X_1}$ from $\mathcal{I}'_{2}$, we obtain the desired information set $\mathcal{I}_{2}$ corresponding to the mutual information $I(X_2;Y|X_1)$ associated with $W_2$. Formally, the three sets are obtained as follows:
\begin{eqnarray}\label{eqn:asymdefinition2nd2}
\begin{cases}
\begin{aligned}
&\text{frozen set: } \mathcal{F}_2=\{i\in[N]:Z(U_2^i|U_2^{1:i-1},Y^{1:N},X_1^{1:N})\geq1-2^{-N^{\beta}}\}\\
&\text{information set: } \mathcal{I}_2=\{i\in[N]:Z(U_2^i|U_2^{1:i-1},Y^{1:N},X_1^{1:N})\leq2^{-N^{\beta}}\text{ and } \\
&\hspace{9em}Z(U_2^i|U_2^{1:i-1},X_1^{1:N})\geq1-2^{-N^{\beta}}\}\\
&\text{shaping set: } \mathcal{S}_2=\left(\mathcal{F}_2\cup\mathcal{I}_2\right)^c.
\end{aligned}
\end{cases}
\end{eqnarray}
\end{lemma}
\begin{proof}
Firstly, we show the proportion of set $\mathcal{S}_{X_1}$ goes to $I(X_1;X_2)$ as $N\to \infty$. Here we define a slightly different set $\mathcal{S}'_{X_1}=\{i\in[N]:Z(U_2^i|U_2^{1:i-1}, X_1^{1:N})\leq2^{-N^{\beta}}\text{ and }Z(U_2^i|U_2^{1:i-1})\geq1-2^{-N^{\beta}}\}$. Suppose we are constructing an asymmetric polar code for the channel from $X_1$ to $X_2$; it is not difficult to find that $\lim_{N\rightarrow\infty}\frac{|\mathcal{S}'_{X_1}|}{N}=I(X_2;X_1)$ by Theorem \ref{theorem:codingtheorem}. Furthermore, by Lemma \ref{lem:CondiReduce}, if $Z(U_2^i|U_2^{1:i-1}, X_1^{1:N})\leq2^{-N^{\beta}}$, we can immediately have $Z(U_2^i|U_2^{1:i-1}, X_1^{1:N}, Y^{1:N})\leq2^{-N^{\beta}}$. Therefore, the difference between the definitions of $\mathcal{S}_{X_1}$ and $\mathcal{S}'_{X_1}$ only lies on $Z(U_2^i|U_2^{1:i-1}, X_1^{1:N})$. Denoting by $\bar{\mathcal{P}}_{X_1}$ the unpolarized set with $2^{-N^{\beta}}  \leq Z(U_2^i|U_2^{1:i-1}, X_1^{1:N})\leq1-2^{-N^{\beta}}$, we have
\begin{eqnarray}
\lim_{N\rightarrow\infty}\left(\frac{|\mathcal{S}_{X_1}|}{N}-\frac{|\mathcal{S}'_{X_1}|}{N}\right) \leq \lim_{N\rightarrow\infty}\frac{|\bar{\mathcal{P}}_{X_1}|}{N}=0.
\end{eqnarray}
As a result, $\lim_{N\rightarrow\infty}\frac{|\mathcal{S}_{X_1}|}{N}=\lim_{N\rightarrow\infty}\frac{|\mathcal{S}'_{X_1}|}{N}=I(X_2;X_1)$.

Secondly, we show that $\mathcal{S}_{X_1} \cup \mathcal{I}_{2}= \mathcal{I}'_{2}$. According to the definitions of $\mathcal{S}_{X_1}$ and $\mathcal{I}_{2}$, we note that $\mathcal{S}_{X_1} \cap \mathcal{I}_2=\emptyset$. By Lemma \ref{lem:CondiReduce}, if $Z(U_2^i|U_2^{1:i-1},X_1^{1:N})\geq1-2^{-N^{\beta}}$, we get $Z(U_2^i|U_2^{1:i-1})\geq1-2^{-N^{\beta}}$ and the difference between the definitions of $\mathcal{S}_{X_1}$ and $\mathcal{I}'_{2}$ only lies on $Z(U_2^i|U_2^{1:i-1}, X_1^{1:N})$. Observe that the union $\mathcal{S}_{X_1} \cup \mathcal{I}_{2}$ would remove the condition on $Z(U_2^i|U_2^{1:i-1}, X_1^{1:N})$, and accordingly we have $\mathcal{S}_{X_1} \cup \mathcal{I}_{2}= \mathcal{I}'_{2}$. It can be also found that the proportion of $\mathcal{I}_2$ goes to $I(X_2;Y|X_1)$ as $N \to \infty$.
\end{proof}

We summarize our main results in the following theorem (see Appendix \ref{appendix1} for the proof):
\begin{theorem}[Coding Theorem for Multilevel Polar Codes]\label{theorem:codingtheoremside}
Consider a polar code with the following encoding and decoding strategies for the channel of the second level $W_2$ with the channel transition probability $P_{Y|X_2,X_1}(y|x_2,x_1)$ shown in \eqref{eqn:transition}.
\begin{itemize}
\item Encoding: Before sending the codeword $x_2^{1:N}=u_2^{1:N}G_N$, the index set $[N]$ are divided into three parts: the frozen set $\mathcal{F}_2$, information set $\mathcal{I}_2$, and shaping set $\mathcal{S}_2$ according to \eqref{eqn:asymdefinition2nd2}. The encoder first places uniformly distributed information bits in $\mathcal{I}_2$. Then the frozen set $\mathcal{F}_2$ is filled with a uniform random sequence which are shared between the encoder and the decoder. The bits in $\mathcal{S}_2$ are generated by a mapping $\phi_{\mathcal{S}_2}\triangleq\{\phi_i\}_{i\in\mathcal{S}_2}$ form a family of randomized mappings $\Phi_{\mathcal{S}_2}$, which yields the following distribution:
\begin{equation}
u_2^i=
\begin{cases}
0 \;\;\;\; \text{with probability } P_{U_2^i|U_2^{1:i-1},X_1^{1:N}}(0|u_2^{1:i-1},x_1^{1:N}),\\
1 \;\;\;\; \text{with probability } P_{U_2^i|U_2^{1:i-1},X_1^{1:N}}(1|u_2^{1:i-1},x_1^{1:N}).
\end{cases}
\label{eqn:encodingside}
\end{equation}
\item Decoding: The decoder receives $y^{1:N}$ and estimates $\hat{u}_2^{1:N}$ based on the previously recovered $x_1^{1:N}$ according to the rule
\begin{equation}
\hat{u}_2^i=
\begin{cases}
u_2^i, \;\;\;\;\;\;\;\;\;\;\;\;\;\;\;\;\text{if } i\in\mathcal{F}_2 \\
\phi_i(\hat{u}_2^{1:i-1}),   \;\;\;\;\; \text{if } i\in\mathcal{S}_2 \\
\underset{u}{\operatorname{argmax}} \; P_{U_2^i|U_2^{1:i-1},X_1^{1:N},Y^{1:N}}(u|\hat{u}_2^{1:i-1},x_1^{1:N},y^{1:N}), \,\text{if } i\in\mathcal{I}_2 \notag\
\end{cases}.
\label{eqn:decoding2}
\end{equation}
\end{itemize}
With the above encoding and decoding, the message rate can be arbitrarily close to $I(Y;X_2|X_1)$ and the expectation of the decoding error probability over the randomized mappings satisfies $E_{\Phi_{\mathcal{S}_2}}[P_e^{SC}(\phi_{\mathcal{S}_2})]= O(2^{-N^{\beta'}})$ for any $0<\beta'<\beta<0.5$, where $\beta$ is used to choose the frozen set, the information set, and the shaping set as in \eqref{eqn:asymdefinition2nd2}.
\end{theorem}

\medskip

Note that probability $P_{U_2^i|U_2^{1:i-1},X_1^{1:N},Y^{1:N}}$ can be calculated by \eqref{eqn:asydec} efficiently, treating $Y$ and $X_1$ (already decoded by the SC decoder at level 1) as the outputs of the asymmetric channel. Again, there exists a deterministic mapping $\phi_{\mathcal{S}_2}$ such that $P_e^{SC}(\phi_{\mathcal{S}_2})= O(2^{-N^{\beta'}})$.

Obviously, Theorem \ref{theorem:codingtheoremside} can be generalized to the construction of a polar code for the channel of the $\ell$-th level $W_{\ell}$. The only difference is that the side information changes from $X_1^{1:N}$ to $X_{1:\ell-1}^{1:N}$. As a result, we can construct a polar code which achieves a rate arbitrarily close to $I(Y;X_\ell|X_{1:\ell})$ with vanishing error probability.

\subsection{Achieving Channel Capacity}
\label{sec:achievingcapacity}

So far, we have constructed polar codes to achieve the capacity of the induced asymmetric channels for all levels. Since the sum capacity of the component channels nearly equals the mutual information $I(Y;X)$, and since we choose a good constellation such that $I(Y;X)\approx \frac{1}{2}\log(1+\SNR)$, we have constructed a lattice code to achieve the capacity of the Gaussian channel.
We summarize the construction in the following theorem:
\begin{theorem}\label{theorem6}
Choose a good constellation with negligible flatness factor $\epsilon_{\Lambda}(\tilde{\sigma})$ as in \cite[Theorem 2]{LingBel13}, and construct a multilevel polar code with $r = O(\log\log N)$ as above. Then, for any SNR, the message rate approaches $\frac{1}{2}\log(1+\SNR)$, while the error probability under multistage SC decoding is bounded by
\begin{eqnarray}
P_{e}^{SC}=  O(2^{-N^{\beta'}}), \quad 0<\beta'<0.5
\label{eqn:errorbound-shap}
\end{eqnarray}
as $N\to \infty$.
\end{theorem}

\begin{remark}
It is simple to generate a transmitted codeword of the proposed scheme. For $n=1$, let
\begin{eqnarray}
\chi = \sum_{\ell=1}^{r}2^{\ell-1}\left[\sum_{i\in\mathcal{I}_{\ell}}u_{\ell}^{i}\mathbf{g}_{i}+\sum_{i\in\mathcal{S}_{\ell}}u_{\ell}^{i}\mathbf{g}_{i}+\sum_{i\in\mathcal{F}_{\ell}}u_{\ell}^{i}\mathbf{g}_{i}\right].
\label{constructionD-finite-power}
\end{eqnarray}
The transmitted codeword $x$ is drawn from $D_{2^r\mathbb{Z}^N+\chi},\sigma_s$. From the proof of Lemma \ref{lem:Iloss}, we know that the probability of choosing a point outside of the interval $[-2^{r-1},2^{r-1}]$ is negligible if $r$ is sufficiently large, which implies there exists only one point in this interval with probability close to $1$. Therefore, one may simply transmit $x=\chi \mod 2^r$, where the modulo operation is applied component-wise with range $(-2^{r-1},2^{r-1}]$.
\end{remark}

Next, we show that such a multilevel polar coding scheme is equivalent to Gaussian shaping over a coset $L+{c}'$ of a polar lattice $L$ for some translate $c'$. In fact, the polar lattice $L$ is exactly constructed from the corresponding symmetrized channels $\tilde{W}_{\ell}$. Recall that the $\ell$-th channel ${W}_{\ell}$ is a BMA channel with the input distribution $P_{X_{\ell}|X_{1:\ell-1}}$ $(1\leq \ell\leq r)$. It is clear that $P_{X_{1:\ell}}(x_{1:\ell})=f_{\sigma_s}(\mathcal{A}_\ell(x_{1:\ell}))/f_{\sigma_s}(\Lambda)$. By Lemma \ref{lem:Asy2sym} and \eqref{eqn:transition}, the transition probability of the symmetrized channel $\tilde{W}_{\ell}$ is
\vspace{-1em}
\begin{eqnarray}
\label{eqn:equivalentsymmetric} &&\hspace{-3em}P_{\tilde{W}_{\ell}}((y,x_{1:\ell-1},x_{\ell}\oplus\tilde{x}_\ell)|\tilde{x}_\ell)\notag\\
&&=P_{Y,X_{1:\ell}}(y,x_{1:\ell}) \notag \\
&&=P_{X_{1:\ell}}(x_{1: \ell})P_{Y|X_{\ell},X_{1:\ell-1}}(y|x_{\ell},x_{1:\ell-1})   \\
&&=\exp\left(-\frac{\|y\|^2}{2(\sigma_s^2+\sigma^2)}\right)\frac{1}{f_{\sigma_s}(\Lambda)}\frac{1}{2\pi\sigma\sigma_s}\sum_{a\in \mathcal{A}_\ell(x_{1:\ell})}\text{exp}\left(-\frac{\|\alpha y-a\|^2}{2\tilde{\sigma}^2}\right). \notag\ 
\end{eqnarray}
Note that the difference between the asymmetric channel \eqref{eqn:transition} and symmetrized channel \eqref{eqn:equivalentsymmetric} is the \textit{a priori} probability $P_{X_{1:\ell}}(x_{1:\ell})=f_{\sigma_s}(\mathcal{A}_\ell(x_{1:\ell}))/f_{\sigma_s}(\Lambda)$. Comparing with the $\Lambda_{\ell-1}/\Lambda_{\ell}$ channel \eqref{eqn:modchannel}, we see that the symmetrized channel \eqref{eqn:equivalentsymmetric} is equivalent to a $\Lambda_{\ell-1}/\Lambda_{\ell}$ channel, since the common terms in front of the sum will be completely cancelled out in the calculation of the likelihood ratio\footnote{Even if $y \in \mathbb{R}^n$ in \eqref{eqn:equivalentsymmetric}, the sum over $\mathcal{A}_\ell(x_{1:\ell})$ is $\Lambda_{\ell}$-periodic. Hence, the likelihood ratio will be the same if one takes $\bar{y} = y \mod \Lambda_{\ell}$ and uses \eqref{eqn:modchannel}}. We summarize the foregoing analysis in the following lemma:

\begin{lemma}[Equivalence lemma]\label{lem:equivalence}
Consider a multilevel lattice code constructed from constellation $D_{\Lambda,\sigma_s}$ for a Gaussian channel with noise variance $\sigma^2$. The $\ell$-th symmetric channel $\tilde{W}_{\ell}$ ($1\leq \ell \leq r$) which is derived from the asymmetric channel $W_{\ell}$ is equivalent to the MMSE-scaled $\Lambda_{\ell-1}/\Lambda_{\ell}$ channel with noise variance $\tilde{\sigma}^2$.
\end{lemma}

Thus, the resultant polar codes for the symmetrized channels are nested, and the polar lattice is AWGN-good for noise variance $\tilde{\sigma}^2$; also, the multistage decoding is performed on the MMSE-scaled signal $\alpha y$ (cf. Lemma \ref{lem:AsymDec}). Since the frozen sets of the polar codes are filled with random bits (rather than all zeros), we actually obtain a coset $L+{c}'$ of the polar lattice, where the shift ${c}'$ accounts for the effects of all random frozen bits. Finally, since we start from $D_{\Lambda,\sigma_s}$, we would obtain $D_{\Lambda^N,\sigma_s}$ without coding; since $L+{c}' \subset \Lambda^N$ by construction, we obtain a discrete Gaussian distribution $D_{L+{c}',\sigma_s}$ over $L+{c}'$.

\begin{remark}
This analysis shows that our proposed scheme is an explicit construction of lattice Gaussian coding introduced in \cite{LingBel13}, which applies Gaussian shaping to an AWGN-good lattice (or its coset). Note that the condition of negligible $\epsilon_{\Lambda}(\tilde{\sigma})$ in Theorem \ref{theorem6} implies negligible capacity $C(\Lambda, \tilde{\sigma}^2)$ of the top lattice in the construction of the AWGN-good lattice in Section \ref{sec:PL}. Again, it is always possible to scale down the top lattice $\Lambda$ such that $\epsilon_{\Lambda}(\tilde{\sigma})$ becomes negligible. Thus, Theorem \ref{theorem6} holds for any SNR, meaning that we have removed the condition $\SNR>e$ required by \cite[Theorem 3]{LingBel13}\footnote{The reason of the condition $\SNR>e$ in \cite{LingBel13} is that a more stringent condition is imposed on the flatness factor of $L$, i.e., $\epsilon_{L}\left(\frac{\sigma_s^2}{\sqrt{\sigma_s^2+\sigma^2}}\right)$ is negligible. Intuitively, for a given lattice $\Lambda$, the flatness factor $\epsilon_{\Lambda}(\sigma)$ decreases as $\sigma$ grows. Namely, the larger $\sigma_s$ is, the smaller $\epsilon_L\left(\frac{\sigma_s^2}{\sqrt{\sigma_s^2+\sigma^2}}\right)$ is. To make $\epsilon_L\left(\frac{\sigma_s^2}{\sqrt{\sigma_s^2+\sigma^2}}\right)$ negligible, $\sigma_s$ can not be arbitrarily small, which then causes an additional condition on the SNR.}. Moreover, if a good constellation of the form $D_{\Lambda-{c},\sigma_s}$ for some shift ${c}$ is used in practice (e.g., a constellation taking values in $\{\pm1, \pm 3, \ldots\}$), the proposed construction holds verbatim.
\end{remark}

\section{Design Examples}
\label{sec:practicaldesign}

In this section, we give design examples of polar lattices with and without the power constraint. The design follows the equal-error-probability rule. Multistage SC decoding is applied. Since the complexity of SC decoding is $O(N\log N)$, the overall decoding complexity is $O(rN\log N)$.

\subsection{Design Examples Without Power Constraint}

\label{sec:example}

\begin{figure}
    \centering
    \includegraphics[width=9cm]{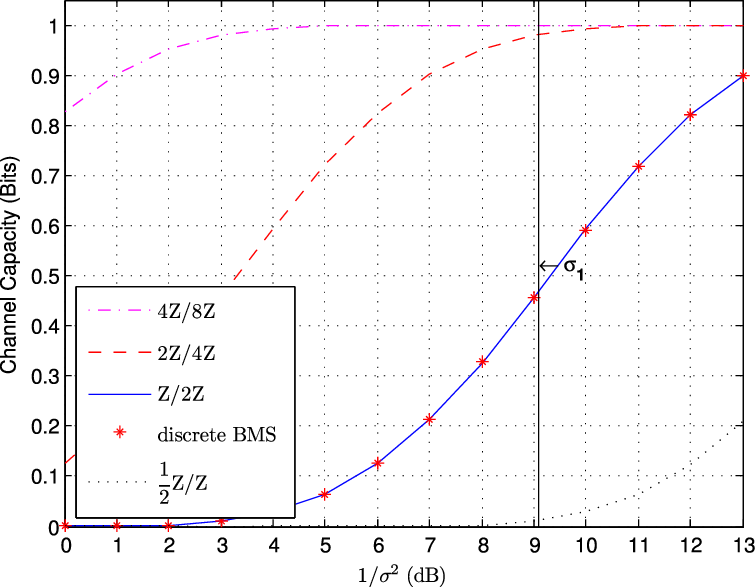}
    \caption{Channel capacity for partition chain $\mathbb{Z}/2\mathbb{Z}/\cdot\cdot\cdot/2^{r}\mathbb{Z}$. The discrete BMS approximation uses the quantization-merging algorithm with $64$ quantization levels.}
    \vspace{-6pt}
    \label{fig:channelcapacity}
\end{figure}

Consider the one-dimensional lattice partition $\mathbb{Z}/2\mathbb{Z}/\cdot\cdot\cdot/2^{r}\mathbb{Z}$. To construct a multilevel lattice, one needs to determine the number of levels of lattice partitions and the actual rates according to the the target error probability for a given noise variance. By the guidelines given in Section \ref{sec:PL}, the effective levels are those which can achieve the target error probability with an actual rate not too close to either $0$ or $1$. Therefore, one can determine the number of effective levels with the help of capacity curves in Fig. \ref{fig:channelcapacity}. For example, for the given noise variance indicated by the straight line in Fig. \ref{fig:channelcapacity}, one may choose partition $\mathbb{Z}/2\mathbb{Z}/4\mathbb{Z}$, i.e., $r=2$, which was indeed suggested in \cite{forney6}.

The multilevel construction and the multistage decoding are shown in Fig. \ref{fig:1D}. For the $\ell$-th level, $\mathbf{g}_{1},\mathbf{g}_{2},\cdots,\mathbf{g}_{k_{\ell}}$ are a set of code generators chosen from the matrix $G_{N}$, and $\sigma_{\ell}$ is the standard deviation of the noise.

\begin{figure}
    \centering
    \includegraphics[width=12cm]{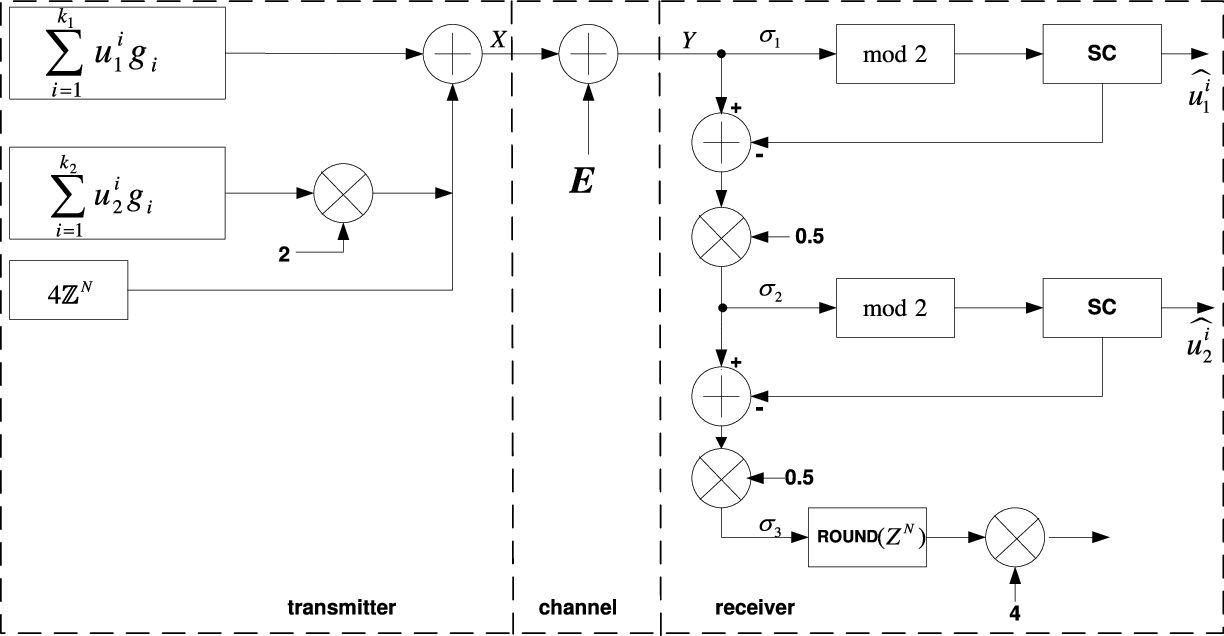}
    \caption{A polar lattice with two levels, where $\sigma_{1}=\sigma$.}
    \vspace{-6pt}
    \label{fig:1D}
\end{figure}

Now, we give an example for length $N=1024$ and target error probability $P_{e}(L,\sigma^{2})= 10^{-5}$. We note that the calculation of $C(\Lambda/\Lambda', \sigma^2)$ and $P_{e}(\Lambda,\sigma^{2})$ can be simplified by the scaling property of the partition channels as shown in the proof of Lemma \ref{lemma:degraded}. For the one-dimensional partition chain, we have $C(4\mathbb{Z}/8\mathbb{Z}, \sigma^2)=C(2\mathbb{Z}/4\mathbb{Z}, (\frac{\sigma}{2})^2)=C(\mathbb{Z}/2\mathbb{Z}, (\frac{\sigma}{4})^2)$, and $P_e(4\mathbb{Z}, \sigma^2)=P_e(2\mathbb{Z}, (\frac{\sigma}{2})^2)=P_e(\mathbb{Z}, (\frac{\sigma}{4})^2)$. Let $\sigma_1=\sigma$, $\sigma_2=\sigma/2$ and $\sigma_3=\sigma/4$ be the equivalent Gaussian noise deviation at the $\ell$'s level with respect to the 1st one.

Since the bottom level is
a $\mathbb{Z}^{N}$ lattice decoder, $\sigma_{3}\approx 0.0845$ for target error probability $\frac{1}{3}\cdot10^{-5}$. For the middle level, $\sigma_{2}=2\cdot\sigma_{3}=0.1690$. From Fig. \ref{fig:channelcapacity}, the channel capacity of the middle
level is $C(\mathbb{Z}/2\mathbb{Z},\sigma_{2}^2)=C(2\mathbb{Z}/4\mathbb{Z},\sigma_{1}^2)=0.9874$. For the top level, $\sigma=\sigma_{1}=0.3380$ and the capacity is $0.5145$. Our goal is to find two polar codes approaching
the respective capacities at block error probabilities $\leq \frac{1}{3}\cdot10^{-5}$ over these binary-input mod-$2$ channels.

For $N=1024$, we
found the first polar code with rate $\frac{k_{1}}{N}=0.23$ for $P_{e}(\mathcal{C}_{1}, \sigma_{1}^{2})\approx \frac{1}{3}\cdot10^{-5}$, and the second polar code with rate $\frac{k_{2}}{N}=0.9$ for $P_{e}(\mathcal{C}_{2}, \sigma_{2}^{2})\approx \frac{1}{3}\cdot10^{-5}$. Recall that the channel in the first level is degraded with respect to the one at the second level according to Lemma \ref{lemma:degraded}, and the two polar codes in this construction turn out to be nested. Thus, the sum rate of component polar codes $R_{\mathcal{C}}=0.23+0.9$, implying a capacity loss $\epsilon_3 = 0.3719$. Meanwhile, the factor $\epsilon_1 = C(\mathbb{Z},0.3380^2)=0.0160$. Therefore, the rate losses at each level are 0.016, 0.285, and 0.087. From \eqref{eqn:minimumVNR},
the logarithmic VNR is given by
\begin{eqnarray}
\log\left(\frac{\gamma_{L}(\sigma)}{2\pi e}\right)\leq2\left(\epsilon_1+\epsilon_3\right)=0.7758,
\label{eqn:gap}
\end{eqnarray}
which is 2.34 dB.
It is seen from Fig. \ref{fig:latticesser} that the estimate 2.34 dB is
very close to the actual gap at $P_{e}(L,\sigma_{1}^{2})\approx 10^{-5}$. This simulation indicates that the gap to the Poltyrev capacity is largely due to the capacity losses of component codes.

A comparison between the polar lattice and Barnes-Wall lattice is also presented in Fig. \ref{fig:1D_bw}. The Barnes-Wall lattices are constructed from Reed-Muller codes at each partition level. By changing the Barnes-Wall rule (base on the hamming weight) to the capacity rule after channel polarization, it can be seen that the performance of the polar lattice is significantly improved.
Thanks to density evolution \cite{mori}, the upper bound $\sum_{i\in \mathcal{A}}\big(Z(W_{N}^{(i)})\big)$ on the block error probability of a polar code with finite length can be calculated numerically. According to \eqref{eqn:errorbound}, we plot the upper bound on the block error probability $P_{e}(L,\sigma^2)$ of the polar lattice in Fig. \ref{fig:1D_bw}, which is quite tight.

\begin{figure}
    \centering
    \includegraphics[width=10cm]{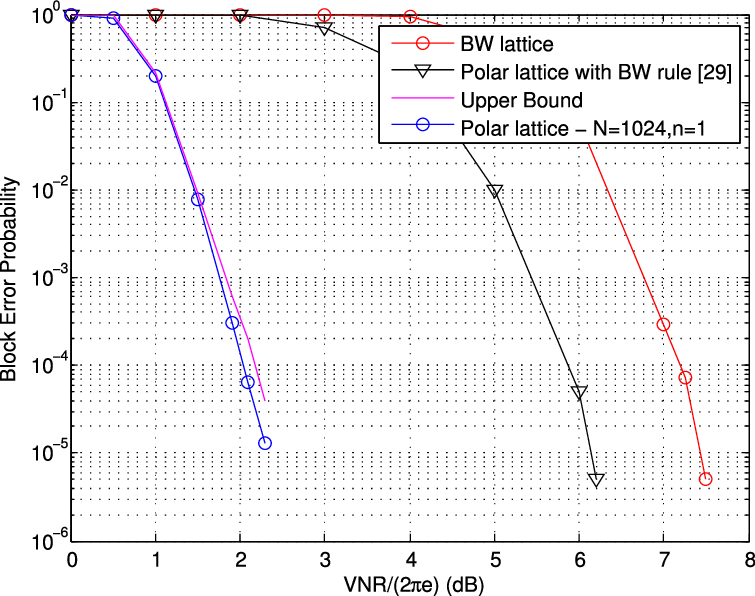}
    \vspace{-12pt}
    \caption{Block error probabilities of polar lattices of length $N=1024$ with multistage decoding.}
    \vspace{-12pt}
    \label{fig:1D_bw}
\end{figure}

We summarize the numerical simulations of polar lattices for infinite constellations as follows. For a given Gaussian noise variance $\sigma^2$ and a target error probability $P_e(L,\sigma^2)$, the first step is to find the smallest $r$ such that $P_e(2^r\mathbb{Z},\sigma^2) \leq \frac{1}{r+1} P_e(L,\sigma^2)$. Then, we need to design $r$ component polar codes with error probability smaller than $\frac{1}{r+1} P_e(L,\sigma^2)$ for the $r$ partition levels, respectively. The proper rate $\frac{k_i}{N}$ of the polar code at the $i$-th level can be estimated by separate numerical simulations or the density evolution technique used in \cite{mori}. By the union bound, the constructed polar lattice is guaranteed to be capable of achieving an error probability lower than $P_e(L,\sigma^2)$.

\begin{figure}[t]
    \centering
    \includegraphics[width=10cm]{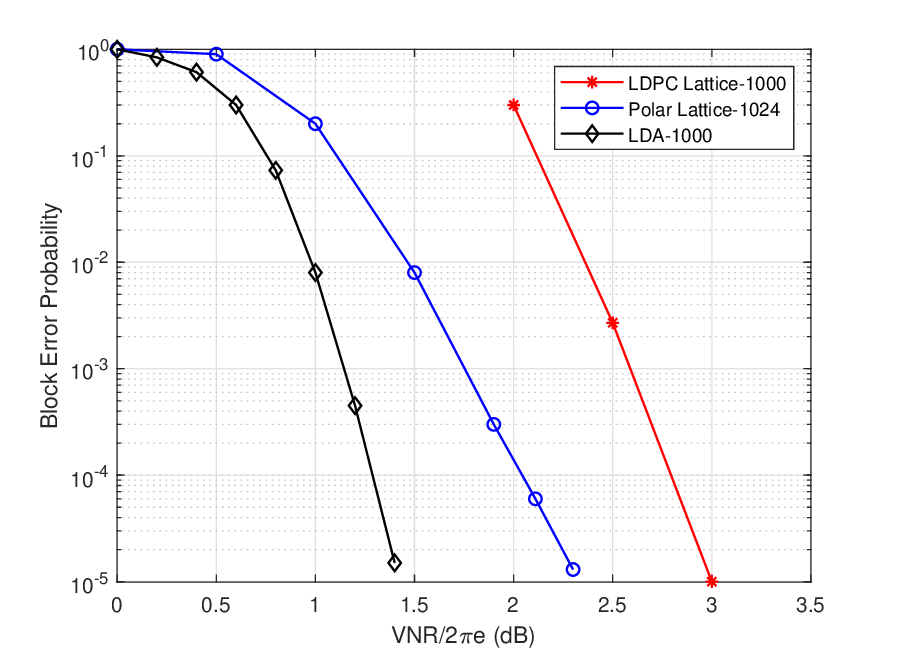}
    \caption{Performance comparison of lattices at dimension around 1000.}
    \label{fig:latticesser}
\end{figure}

Performance comparison of competing lattices approaching the Poltyrev capacity is presented in Fig. \ref{fig:latticesser}, at dimension around 1000. The polar lattice used here is constructed from the aforementioned one-dimensional lattice partition ($N=1024, n=1$). The simulation curves of other lattices are taken from their corresponding papers. Among the three types of lattices compared, the LDPC lattice \cite{ldpclattice} has the weakest performance. The LDA lattice \cite{PietroZB18} has better performance than the polar lattice, at the expense of higher decoding complexity $O(p^2N\log N)$ if $p$-ary LDPC codes are employed. Assuming $p\approx 2^r$, it would require complexity $O(2^{2r}N\log N)$, compared to $O(r N\log N)$ of the polar lattice. The LDLC lattice is not included in this comparison because of lack of block error probabilities in \cite{ldlc}. In contrast to the polar lattice and LDA lattice, analytic results of the LDLC are not available; therefore, they are less understood in theory. It is worth pointing out that the plain polar codes used in polar lattice can be optimized in several aspects: for example, to use a better kernel, list decoding, or even a soft-output decoding algorithm. We leave such improvements of polar lattices to future work.

\subsection{Design Examples With Power Constraint}
\label{sbsec:ExamplewithS}
\begin{figure}[h]
    \centering
    \includegraphics[width=9cm]{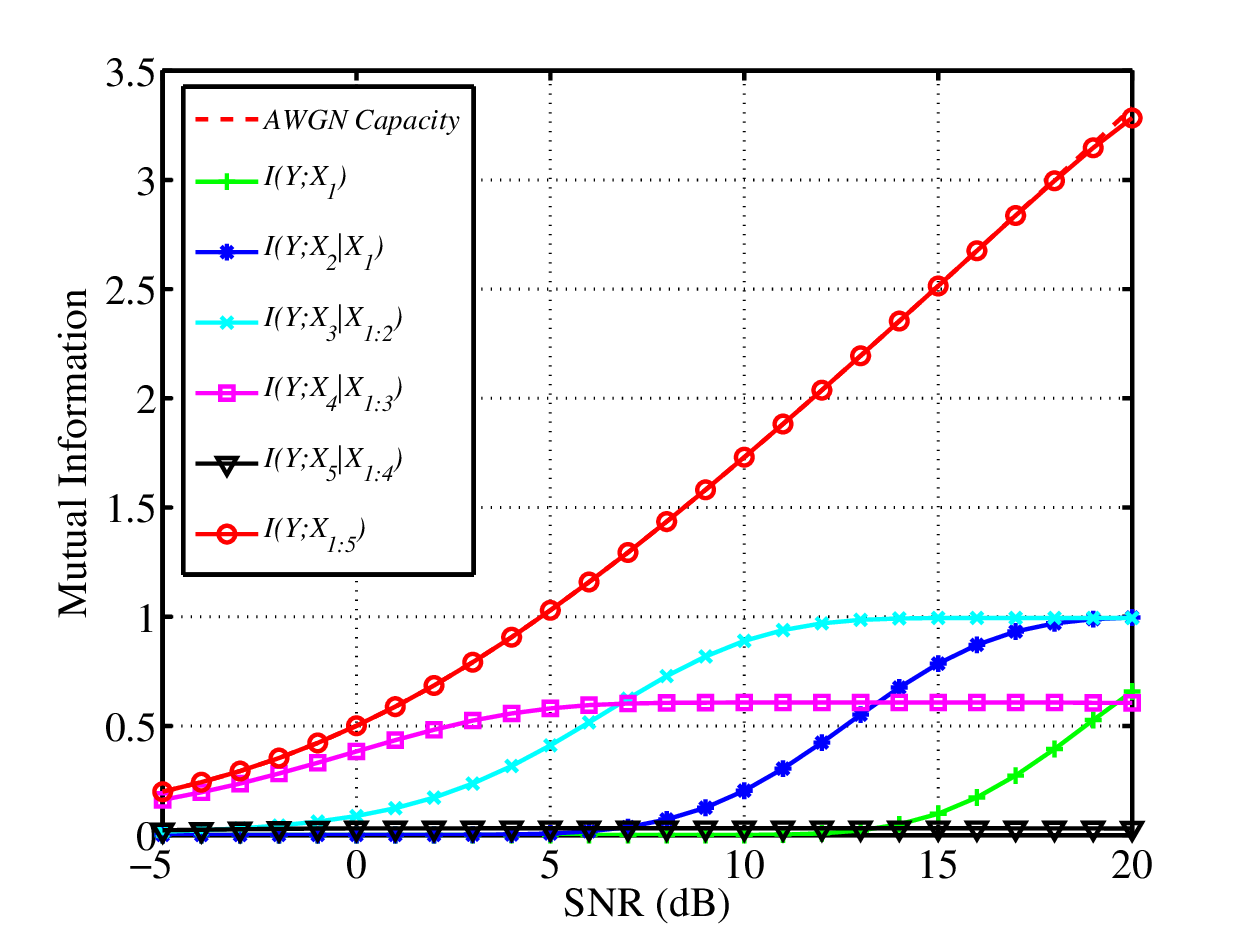}
    \vspace{-2em}
    \caption{Channel capacity for each level as a function of SNR.}
    \label{fig:Record2}
\end{figure}

To satisfy the power constraint, we use discrete lattice distribution $D_{\mathbb{Z},\sigma_s}$ for shaping. The mutual information $I(Y;X_\ell|X_{1:\ell-1})$ at each level for different SNRs is shown in Figure \ref{fig:Record2}. We can see that for partition $\mathbb{Z}/2\mathbb{Z}/...$,  five levels are enough to achieve the AWGN channel capacity for SNR ranging from $-5$ dB to $20$ dB. Note that the actual number of required levels depends on the SNR: a smaller number of levels are enough for low SNRs, while a larger number of levels is required for high SNRs (to support higher rates).

\begin{figure}[H]
    \centering
    \includegraphics[width=9cm]{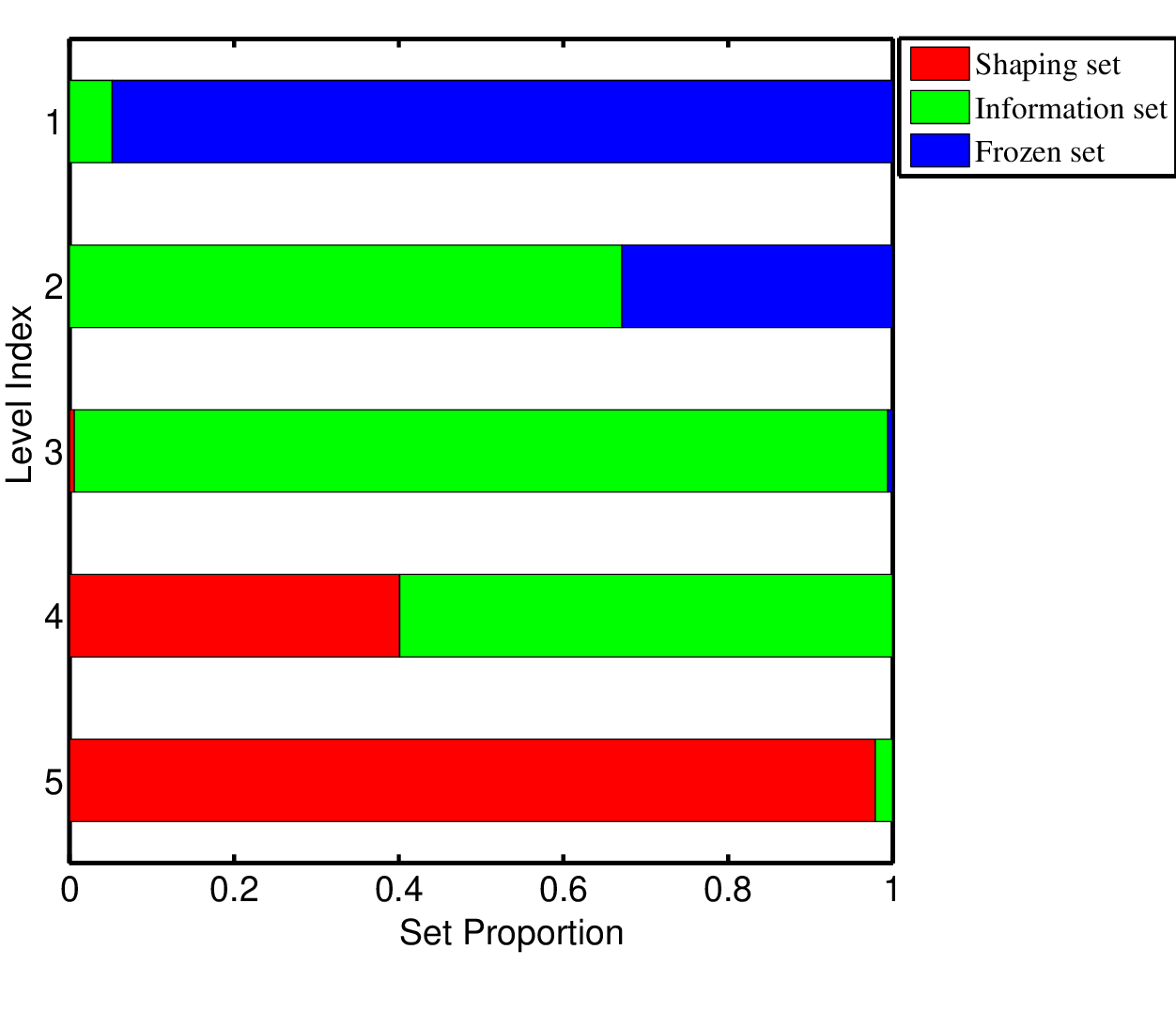}
    \vspace{-2em}
    \caption{The proportions of the shaping set, information set, and frozen set on each level when $N=2^{16}$ and $\SNR=15$ dB.}
    \label{fig:Record10}
\end{figure}

For each level, we estimate a lower bound on the code rate for block error probability $1 \times 10^{-5}$. This is done by calculating an upper bound on the block error probability of the polar code, using the Bhattacharyya parameter. With this target error probability, the assignments of bits to the information, shaping and frozen sets on different levels are shown in Figure \ref{fig:Record10} for $\SNR=15$ dB and $N=2^{16}$. In fact, $X_{1}$ and $X_{2}$ are nearly uniform such that there is no need for shaping on the first two levels (these levels actually correspond to the AWGN-good lattice). The third level channel is very clean, and most bits are information bits. In contrast, the fifth level is mostly for shaping; since its message rate is already small, adding another level clearly would not contribute to the overall rate of the lattice code. Finally, lower bounds on the rates achieved by polar lattices with various block lengths are shown in Figure \ref{fig:Record11}. We note that the gap to the channel capacity diminishes as $N$ increases, and it is only about $0.1$ bits/dimension when $N=2^{20}$.
\begin{figure}[H]
    \centering
    \includegraphics[width=9cm]{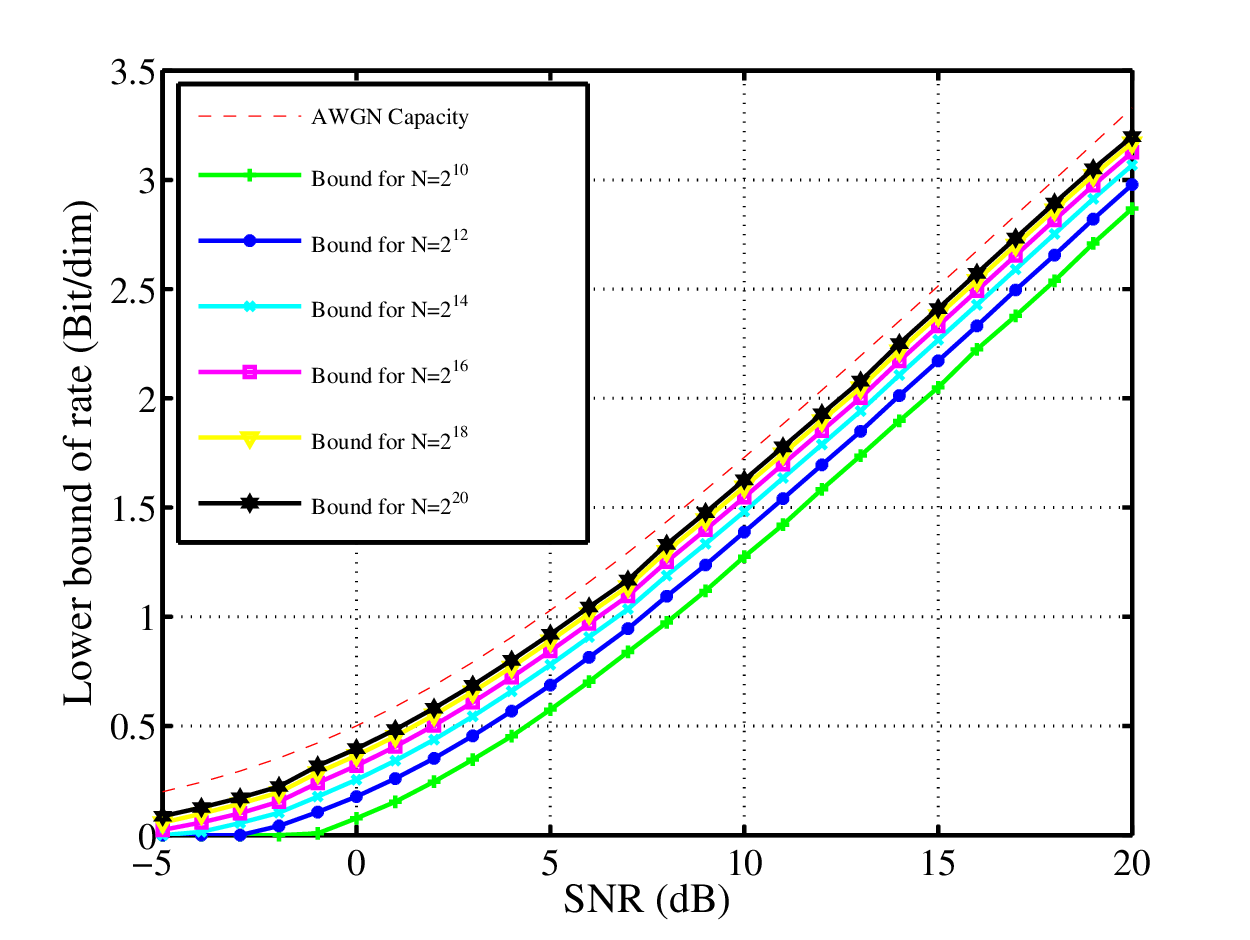}
     \vspace{-2em}
    \caption{Lower bounds on the rates achieved by polar lattices with block error probability $5\times 10^{-5}$ for block lengths $2^{10},..., 2^{20}$.}
    \label{fig:Record11}
\end{figure}

\section{Conclusions}

In this paper, we have constructed polar lattices to approach the capacity of the power-constrained Gaussian channel. The construction is based on a combination of channel polarization and source polarization. Without shaping, the constructed polar lattices are AWGN-good. The Gaussian shaping on a polar lattice deals with the power constraint but is technically more involved. Our shaping approach is different from the standard Voronoi shaping which involves a quantization-good lattice \cite{zamir}. The proposed Gaussian shaping does not require such a quantization-good lattice any more. The overall scheme is explicit and efficient, featuring quasi-linear complexity.

%

\appendices


\section{Proof of Lemma \ref{lem:numberoflevels}}\label{appendixlevels}

\begin{proof}
For this purpose, we assume $\Lambda=a\mathbb{Z}^n$ and $\Lambda'=b\mathbb{Z}^n$ where $a, b$ are scaling parameters to be estimated. We note that for all partition chains in  \cite{forney6}, this is always possible: if the bottom lattice does not take the form of $b\mathbb{Z}^n$, one may simply further extend the partition chain (which will lead to an upper bound on $r$).

We firstly note that the flatness factor $\epsilon_{\Lambda}(\sigma)$ can be made arbitrarily small by scaling down the top lattice $\Lambda$. To see this, we recall that $\epsilon_{\Lambda}(\sigma_e)\leq [1+\epsilon_{\Lambda_0}(\sigma_e)]^n-1$ \cite[Lemma 3]{LingBel13} where $\Lambda_0= a \mathbb{Z}$ for the afore-mentioned scaling factor $a$.


Let $\Lambda_0^*=\frac{1}{a}\mathbb{Z}$ be the dual lattice of $\Lambda_0$. By \cite[Corollary 1]{cong2}, we have
\begin{eqnarray}\label{eqn:epsilon1Z}
\begin{aligned}
\epsilon_{\Lambda_0}(\sigma)&=\Theta_{\Lambda_0^*}(2\pi \sigma^2)-1 \\
&=\sum_{\lambda \in \Lambda_0^*}\exp(-2\pi^2\sigma^2|\lambda|^2 )-1\\
&=2\sum_{\lambda \in \frac{1}{a}\mathbb{N}} \exp(-2\pi^2\sigma^2|\lambda|^2 ) \\
&\leq\frac{2\exp(-2 \pi^2\sigma^2\frac{1}{a^2})}{1-\exp(-2 \pi^2\sigma^2\frac{3}{a^2})}\\
&\leq4\exp(-2 \pi^2\sigma^2\frac{1}{a^2}) \quad \text{for sufficiently small } a.
\end{aligned}
\end{eqnarray}
Therefore, letting $\frac{1}{a}=\sqrt{\frac{N}{2\pi^2 \sigma^2}}$, we have $\epsilon_{\Lambda_0}(\sigma)=O(e^{-{N}})$ and hence $\epsilon_{\Lambda}(\sigma)=O(e^{-{N}})$ for fixed $n$.

Secondly, by the union bound, the error probability of the bottom lattice $\Lambda'$ is upper-bounded by
\begin{eqnarray}
P_{e}(\Lambda',\sigma^{2})
\leq nQ\left(\frac{b}{2\sigma}\right) \leq n e^{-\frac{b^2}{8\sigma^{2}}} \notag
\end{eqnarray}
where we apply the Chernoff bound on the Q-function.
We can obtain
\begin{eqnarray}
P_{e}(\Lambda',\sigma^{2})=O(e^{-{N}})
\notag\
\end{eqnarray}
by choosing $b=\sqrt{8\sigma^2 N}$ for fixed $n$.

For a binary lattice partition, we have $(b/a)^n=2^{r}$. Thus, we conclude that
\begin{eqnarray}
r = n \log\left(\frac{b}{a}\right)= n \log \left(\frac{2}{\pi} N\right) \leq n\log N = O(\log N). \notag\
\end{eqnarray}
\end{proof}

\section{Proof of Lemma \ref{lemma:degraded}}\label{Appendixdegraded}

\begin{proof}
By the self-similarity of the lattice partition chain, we can scale a $\Lambda_{\ell-1}/\Lambda_{\ell}$ channel to a $\Lambda_{\ell}/\Lambda_{\ell+1}$ channel by multiplying the output of a $\Lambda_{\ell-1}/\Lambda_{\ell}$ channel with $T$. Since $T=\alpha V$ for some scale factor $\alpha>0$ and orthogonal matrix $V$, the Gaussian noise for each dimension is still independent of each other and the noise variance per dimension is increased after the scaling. Therefore, a $\Lambda_{\ell-1}/\Lambda_{\ell}$ channel is stochastically equivalent to a $\Lambda_{\ell}/\Lambda_{\ell+1}$ channel with a larger Gaussian noise variance per dimension. For our design examples, a $\mathbb{Z}/2\mathbb{Z}$ channel with Gaussian noise variance $\sigma^2$ is equivalent to a $2\mathbb{Z}/4\mathbb{Z}$ channel with Gaussian noise variance $4\sigma^2$, and a $\mathbb{Z}^{2}/R\mathbb{Z}^{2}$ channel with noise variance $\sigma^2$ per dimension is equivalent to a $R\mathbb{Z}^{2}/2\mathbb{Z}^{2}$ channel with noise variance $2\sigma^2$ per dimension. Then our task is to prove that a $\Lambda_{\ell}/\Lambda_{\ell+1}$ channel with noise variance $\sigma_2^2$ is degraded with respect to a $\Lambda_{\ell}/\Lambda_{\ell+1}$ channel with noise variance $\sigma_1^2$ if $\sigma_1^2 \leq \sigma_2^2$.

To see the channel degradation, we construct an intermediate channel with input in $\mathcal{R}(\Lambda_{\ell+1})$ and a mod-$\Lambda_{\ell+1}$ operation at the receiver's front end. The noise variance of this mod-$\Lambda_{\ell+1}$ channel is given by $\sigma_2^2-\sigma_1^2$ per dimension. By the property $[X+Y] \mod \Lambda_{\ell+1}= \Big[X \mod \Lambda_{\ell+1}+Y \Big] \mod \Lambda_{\ell+1}$, we can find that the concatenated channel that consists of a $\Lambda_\ell/\Lambda_{\ell+1}$ channel with noise variance $\sigma_1^2$ followed by the mentioned intermediate channel is stochastically
equivalent to a $\Lambda_\ell/\Lambda_{\ell+1}$ channel with noise variance $\sigma_2^2$, in the sense that the channel transition probability density functions of the two channels for any given input and output are equivalent. This relationship is depicted in Fig. \ref{fig:degradedproof}.
According to Definition \ref{deft:degradation}, the proof is completed.

\begin{figure}[H]
    \centering
    \includegraphics[width=10cm]{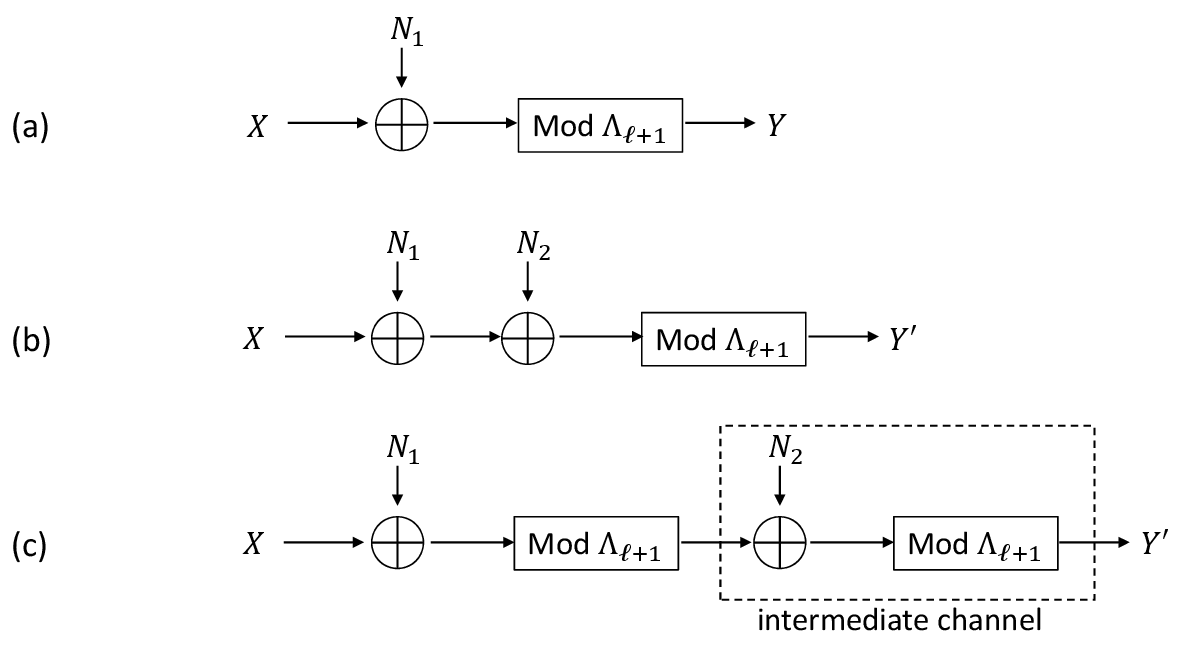}
    \caption{Let $X \in \mathcal{R}(\Lambda_{\ell+1})$ denote the channel input. Let $N_1$ and $N_2$ denote two independent additive Gaussian noise with variances $\sigma_1^2$ and $\sigma_2^2-\sigma_1^2$, respectively. Clearly, the two $\Lambda_{\ell}/\Lambda_{\ell+1}$ channels with noise variances $\sigma_1^2$ and $\sigma_2^2$ can be described by channel (a) and (b), respectively. By the property of modulo operation, channel (b) is equivalent to channel (c), which is a concatenated channel made by concatenating channel (a) with an intermediate mod-$\Lambda_{\ell+1}$ channel.}
    \label{fig:degradedproof}
\end{figure}

\end{proof}

\section{Proof of Lemma \ref{lem:Iloss}}\label{AppendixIloss}
\begin{proof}
For convenience we consider a one-dimensional partition chain $\mathbb{Z}/2\mathbb{Z}/\cdots$. The proof can be extended to the multi-dimensional case by sandwiching the partition in $\mathbb{Z}^n/2\mathbb{Z}^n/\cdots$, which reduces to the one-dimensional case.

For level $r$, the selected coset $\mathcal{A}_r$ can be written as $x_1+\cdot\cdot\cdot 2^{r-1}x_r+2^r\mathbb{Z}$. Clearly, $\mathcal{A}_r$ is a subset of $\mathcal{A}_{r-1}$. Let $\lambda_1$ and $\lambda_2$ denote the two lattice points with smallest norm in set $\mathcal{A}_{r-1}$. Without loss of generality, we assume $\lambda_1 \leq 0 \leq \lambda_2$ and $|\lambda_1| \leq |\lambda_2|$. Observe that $\lambda_2-\lambda_1=2^{r-1}$.
Assume $2^{r-1}= 3T\sigma_s$, and $T=\delta \log N$ for some positive constant $\delta$, then $\lambda_1$ and $\lambda_2$ cannot be in the interval $[-T\sigma_s, T\sigma_s]$ simultaneously. We consider two cases.

Case I: If the two points are both outside of $[-T\sigma_s, T\sigma_s]$, then we have
\begin{eqnarray}
\begin{aligned}
P(\mathcal{A}_{r-1})=\frac{f_{\sigma_s}(\mathcal{A}_{r-1}(x_{1:r-1}))}{f_{\sigma_s}(\mathbb{Z})}&<\frac{\frac{1}{\sqrt{2\pi} \sigma_s}\sum\limits_{x\in2^{r-1}\mathbb{Z}}\exp(-\frac{(x+\lambda_1)^2}{2\sigma_s^2})}{\frac{1}{\sqrt{2\pi} \sigma_s}} \\ \notag\
&\leq 2 \sum\limits_{x\in2^{r-1}\mathbb{Z}_-}\exp\left(-\frac{(x+\lambda_1)^2}{2\sigma_s^2}\right) \\
&\leq 2 \sum\limits_{x\in2^{r-1}\mathbb{Z}_-}\exp\left(-\frac{x^2+\lambda_1^2}{2\sigma_s^2}\right) \\
&\stackrel{(a)}\leq 2 \exp\left(-\frac{\lambda_1^2}{2\sigma_s^2}\right)\sum\limits_{n\in\mathbb{Z}_-} \exp\left(n\frac{(2^{r-1})^2}{2\sigma_s^2}\right) \\
&\leq 2 \frac{\exp (-\frac{\lambda_1^2}{2\sigma_s^2})}{1-\exp(-\frac{(2^{r-1})^2}{2\sigma_s^2})} \leq 2 \frac{\exp (-\frac{T^2}{2})}{1-\exp(-\frac{9T^2}{2})},
\end{aligned}
\end{eqnarray}
where $\mathbb{Z}_-$ represents all non-positive integers and trivial bound $n \leq n^2$ for $n\in \mathbb{Z}$ is applied in step (a). This means $P(\mathcal{A}_{r-1})$ roughly scales as $\frac{1}{N^{\log N}}$, so $P(\mathcal{A}_{r-1}) = O(\frac{1}{N^c})$ for any constant $c>0$.

Case II: The point $\lambda_1$ is in the interval $[-T\sigma_s, T\sigma_s]$ while $\lambda_2$ lies outside. Without loss of generality, we assume that the two cosets corresponding to $x_r=0$ and $x_r=1$ are $\lambda_1+2^r \mathbb{Z}$ and $\lambda_2+2^r \mathbb{Z}$, respectively. Then we have
\begin{eqnarray}
\begin{aligned}
\frac{P(x_r=0|x_{1:r-1})}{P(x_r=1|x_{1:r-1})}&=\frac{\sum\limits_{x\in2^r\mathbb{Z}} \exp(-\frac{(x+\lambda_1)^2}{2\sigma_s^2})}{\sum\limits_{x\in2^r\mathbb{Z}} \exp(-\frac{(x+\lambda_2)^2}{2\sigma_s^2})} \\ \notag\
&\geq \frac{\exp (-\frac{\lambda_1^2}{2\sigma_s^2})}{2\sum\limits_{x\in2^r\mathbb{Z}_+} \exp (-\frac{(x+\lambda_2)^2}{2\sigma_s^2})} \\
&\geq \frac{\exp(-\frac{\lambda_1^2}{2\sigma_s^2})}{2\cdot\exp(-\frac{\lambda_2^2}{2\sigma_s^2})}  \left(1-\exp\left(-\frac{2^{2r}}{2\sigma_s^2}\right)\right),
\end{aligned}
\end{eqnarray}
where $\mathbb{Z}_+$ represents all non-negative integers. Since $\lambda_2-\lambda_1=2^{r-1}=3T\sigma_s$ and $\lambda_2+\lambda_1 \geq T\sigma_s$, for any $T>1$, we can obtain
\begin{eqnarray}
\begin{aligned}
\frac{P(x_r=0|x_{1:r-1})}{P(x_r=1|x_{1:r-1})} &\geq \frac{1}{2}\exp\left(\frac{3}{2}T^2\right)(1-\exp(-18T^2)) \\ \notag\
&\geq \frac{1}{4}\exp\left(\frac{3}{2}T^2\right)= \frac{1}{4}\exp\left(\frac{3}{2}\delta^2 \log^2 N\right).
\end{aligned}
\end{eqnarray}
Assuming that $\frac{1}{4}\exp(\frac{3}{2}\delta^2 \log^2 N)=M$, we can get $P(x_r=0|x_{1:r-1})\geq \frac{M}{M+1}$ and $P(x_r=1|x_{1:r-1})\leq \frac{1}{M+1}$. Then we have,
\begin{eqnarray}
I(Y;X_{r}|X_{1:r-1}) \leq H(X_{r}|X_{1:r-1}) \leq h_2\left(\frac{1}{M+1}\right), \notag\
\end{eqnarray}
where $h_2(p)=p\text{log}(\frac{1}{p})+(1-p)\text{log}(\frac{1}{1-p})$ denotes the binary entropy function. By the relationship $\text{ln}(x) \leq \frac{x-1}{\sqrt{x}}$ when $x \geq 1$, we finally have
\begin{eqnarray}
I(Y;X_{r}|X_{1:r-1}) \leq \log(e)\left(\frac{1}{\sqrt{M}}+\frac{1}{M}\right)=\log(e)\left(\frac{2}{\exp(\delta_12^{2r})}+ \frac{4}{\exp(\delta_22^{2r})}\right), \notag\
\end{eqnarray}
where $\delta_1$ and $\delta_2$ are two positive constants. Therefore, there exists $r= O(\log \log N)$ such that $I(Y;X_{r}|X_{1:r-1})\rightarrow 0$ as $N$ increases, and $\sum_{\ell \geq r} I(Y;X_{\ell}|X_{1:\ell-1}) = O(\frac{1}{N})$.

To see this, let $\frac{2}{\exp(\delta_12^{2r})}=\frac{1}{N^c}$ for any constant $c>2$. From this we derive $r=\frac{1}{2}\log\frac{\log(2N^c)}{\delta_1}= O(\log\log N)$. Then for sufficiently large $N$, $I(Y;X_{r}|X_{1:r-1}) \leq \log(e)\frac{2}{N^c}$ and
\begin{eqnarray*}
  \sum_{\ell \geq r} I(Y;X_{\ell}|X_{1:\ell-1})  &\leq& \sum_{n \geq N} \log(e)\frac{2}{n^c} \\
   &\leq& \log(e)\frac{2}{N} \sum_{n \geq N} \frac{1}{n^{c-1}} \\
   &\leq&  \log(e)\frac{2}{N} \sum_{n \geq 1} \frac{1}{n^{c-1}} \\
   &\leq&  \log(e)\frac{2}{N} \zeta(c-1)
\end{eqnarray*}
where $\zeta(x)$ denotes the Riemann zeta function, which converges for any real $x>1$.

Finally, applying the total probability theorem to both cases and noting that Case I also happens with probability $O(\frac{1}{N^c})$ for any $c>2$, we conclude that $I(Y;X_{r}|X_{1:r-1}) = O(\frac{1}{N^c})$ for $c>2$, hence $\sum_{\ell \geq r} I(Y;X_{\ell}|X_{1:\ell-1}) = O(\frac{1}{N})$.

\end{proof}

\section{Proof of Theorem \ref{theorem:codingtheorem}} \label{appendix0}
\begin{proof}
Let $\mathcal{E}_i$ denote the set of pairs of $u^{1:N}$ and $y^{1:N}$ such that decoding error occurs at the $i$-th bit, then the block decoding error event is given by $\mathcal{E} \equiv \bigcup_{i \in \mathcal{I}} \mathcal{E}_i$. According to our encoding scheme, each codeword $u^{1:N}$ appears with probability
\begin{equation}
2^{-(|\mathcal{I}|+|\mathcal{F}|)} \prod_{i \in \mathcal{S}} P_{U^i|U^{1:i-1}}(u^i|u^{1:i-1}). \notag\
\end{equation}
Then the expectation of decoding error probability over all random mapping is expressed as
\begin{eqnarray}
\begin{aligned}
E[P_e]=\sum_{u^{1:N},y^{1:N}} &2^{-(|\mathcal{I}|+|\mathcal{F}|)} (\prod_{i \in \mathcal{S}} P_{U^i|U^{1:i-1}}(u^i|u^{1:i-1})) \\ \notag\
&\cdot P_{Y^{1:N}|U^{1:N}}(y^{1:N}|u^{1:N}) \mathds{1}[(u^{1:N},y^{1:N})\in \mathcal{E}].
\end{aligned}
\end{eqnarray}
Now we define the probability distribution $Q_{U^{1:N},Y^{1:N}}$ as
\begin{eqnarray}
Q_{U^{1:N},Y^{1:N}}(u^{1:N},y^{1:N})=2^{-(|\mathcal{I}|+|\mathcal{F}|)} (\prod_{i \in \mathcal{S}} P_{U^i|U^{1:i-1}}(u^i|u^{1:i-1}))P_{Y^{1:N}|U^{1:N}}(y^{1:N}|u^{1:N}). \notag\
\end{eqnarray}
Then the variational distance between $Q_{U^{1:N},Y^{1:N}}$ and $P_{U^{1:N},Y^{1:N}}$ can be bounded as
{\allowdisplaybreaks\begin{eqnarray}\label{eqn:Errstep1}
\begin{aligned}
&\hspace{-2em}2\|Q_{U^{1:N},Y^{1:N}}-P_{U^{1:N},Y^{1:N}}\|= \sum_{u^{1:N},y^{1:N}} |Q(u^{1:N},y^{1:N})-P(u^{1:N},y^{1:N})|\\
&\stackrel{(a)}=\sum_{u^{1:N},y^{1:N}}|\sum_i(Q(u^i|u^{1:i-1})-P(u^i|u^{1:i-1}))(\prod_{j=1}^{i-1}P(u^i|u^{1:i-1}))(\prod_{j=i+1}^{N}Q(u^i|u^{1:i-1}))Q(y^{1:N}|u^{1:N})|\\
&\leq\sum_{i\in\mathcal{I}\cup\mathcal{F}}\sum_{u^{1:N},y^{1:N}}|Q(u^i|u^{1:i-1})-P(u^i|u^{1:i-1})|(\prod_{j=1}^{i-1}P(u^i|u^{1:i-1}))(\prod_{j=i+1}^{N}Q(u^i|u^{1:i-1}))Q(y^{1:N}|u^{1:N})\\
&=\sum_{i\in\mathcal{I}\cup\mathcal{F}} \sum_{u^{1:i-1}} 2P(u^{1:i-1})\|Q_{U^i|U^{1:i-1}=u^{1:i-1}}-P_{U^i|U^{1:i-1}=u^{1:i-1}}\|\\
&\stackrel{(b)}\leq \sum_{i\in\mathcal{I}\cup\mathcal{F}} \sum_{u^{1:i-1}} P(u^{1:i-1}) \sqrt{2\text{ln}2 D(P_{U^i|U^{1:i-1}=u^{1:i-1}}\|Q_{U^i|U^{1:i-1}=u^{1:i-1}})}\\
&\leq\sum_{i\in\mathcal{I}\cup\mathcal{F}} \sqrt{2\text{ln}2 \sum_{u^{1;i-1}}P(u^{1:i-1}) D(P_{U^i|U^{1:i-1}=u^{1:i-1}}\|Q_{U^i|U^{1:i-1}=u^{1:i-1}})}\\
&\leq\sum_{i\in\mathcal{I}\cup\mathcal{F}} \sqrt{2\text{ln}2 D(P_{U^i|U^{1:i-1}}||Q_{U^i|U^{1:i-1}})}\\
&\leq\sum_{i\in\mathcal{I}} \sqrt{2\text{ln}2(1-H(U^i|U^{1:i-1}))}+\sum_{i\in\mathcal{F}} \sqrt{2\text{ln}2(1-H(U^i|U^{1:i-1}))}\\
&\leq\sum_{i\in\mathcal{I}} \sqrt{2\text{ln}2(1-Z(U^i|U^{1:i-1})^2)}+\sum_{i\in\mathcal{F}} \sqrt{2\text{ln}2(1-Z(U^i|U^{1:i-1},Y^{1:N})^2)}\\
&\leq2N\sqrt{4\text{ln}2\cdot2^{-N^\beta}},
\end{aligned}
\end{eqnarray}}
where equality $(a)$ follows from \cite[Equation (56)]{aspolarcodes} and $Q(y^{1:N}|u^{1:N})=P(y^{1:N}|u^{1:N})$. $D(\cdot||\cdot)$ in the inequality $(b)$ is the relative entropy, and this inequality holds because of the Pinsker's inequality.
Then we have
\begin{eqnarray}\label{eqn:Errstep2}
\begin{aligned}
E[P_e]&=Q_{U^{1:N},Y^{1^N}}(\mathcal{E})\\
&\leq\|Q_{U^{1:N},Y^{1:N}}-P_{U^{1:N},Y^{1:N}}\|+P_{U^{1:N},Y^{1:N}}(\mathcal{E})\\
&\leq\|Q_{U^{1:N},Y^{1:N}}-P_{U^{1:N},Y^{1:N}}\|+\sum_{i\in \mathcal{I}} P_{U^{1:N},Y^{1:N}}(\mathcal{E}_i),
\end{aligned}
\end{eqnarray}
where
\begin{eqnarray}
\begin{aligned}
P_{U^{1:N},Y^{1:N}}(\mathcal{E}_i)&\leq\sum_{u^{1:N},y^{1:N}}P(u^{1;i-1},y^{1:N})P(u^i|u^{1:i-1},y^{1:N})\cdot\mathds{1}[P(u^i|u^{1:i-1},y^{1:N})\leq P(u^i\oplus1|u^{1:i-1},y^{1:N})]\\ \notag\
&\leq \sum_{u^{1:N},y^{1:N}}P(u^{1;i-1},y^{1:N})P(u^i|u^{1:i-1},y^{1:N})\sqrt{\frac{P(u^i\oplus1|u^{1:i-1},y^{1:N})}{P(u^i|u^{1:i-1},y^{1:N})}}\\
&=Z(U^i|U^{1:i-1},Y^{1:N})\leq 2^{-N^\beta}.
\end{aligned}
\end{eqnarray}
From \eqref{eqn:Errstep1} and \eqref{eqn:Errstep2}, we have $E[P_e]\leq 2N\sqrt{4\text{ln}2\cdot2^{-N^\beta}}+N2^{-N^\beta}=N2^{-N^{\beta'}}$ for any $\beta'<\beta<0.5$.
\end{proof}

\section{Proof of Theorem \ref{theorem:codingtheoremside}}\label{appendix1}
\begin{proof}
Let $\mathcal{E}_i$ denote the set of triples of $u_2^{1:N}$, $x_1^{1:N}$ and $y^{1:N}$ such that decoding error occurs at the $i$-th bit, then the block decoding error event is given by $\mathcal{E} \equiv \bigcup_{i \in \mathcal{I}} \mathcal{E}_i$. According to our encoding scheme, each codeword $u_2^{1:N}$ appears with probability
\begin{eqnarray}
2^{-(|\mathcal{I}_2|+|\mathcal{F}_{2}|)} \prod_{i \in \mathcal{S}_2} P_{U_2^i|U_2^{1:i-1},X_1^{1:N}}(u_2^i|u_2^{1:i-1},x_1^{1:N}). \notag\
\end{eqnarray}
Then the expectation of decoding error probability over all random mapping is expressed as
\begin{eqnarray}
\begin{aligned}
E[P_e]&=&\sum_{u_2^{1:N},x_1^{1:N},y^{1:N}} 2^{-(|\mathcal{I}_2|+|\mathcal{F}_2|)} (\prod_{i \in \mathcal{S}_2} P_{U_2^i|U_2^{1:i-1},X_1^{1:N}}(u_2^i|u_2^{1:i-1},x_1^{1:N})) \\ \notag\
&& \cdot P_{Y^{1:N},X_1^{1:N}|U_2^{1:N}}(y^{1:N},x_1^{1:N}|u_2^{1:N}) \mathds{1}[(u_2^{1:N},x_1^{1:N},y^{1:N})\in \mathcal{E}].
\end{aligned}
\end{eqnarray}
Now we define the probability distribution $Q_{U_2^{1:N},X_1^{1:N},Y^{1:N}}$ as
\begin{eqnarray}
\begin{aligned}
Q_{U_2^{1:N},X_1^{1:N},Y^{1:N}}(u_2^{1:N},x_1^{1:N},y^{1:N})=&2^{-(|\mathcal{I}_2|+|\mathcal{F}_{2}|)}\cdot Q_{X_1^{1:N}}(x_1^{1:N}) \\
&(\prod_{i \in \mathcal{S}_2} P_{U_2^i|U_2^{1:i-1},X_1^{1:N}}(u_2^i|u_2^{1:i-1},x_1^{1:N})) \cdot P_{Y^{1:N}|X_1^{1:N},U_2^{1:N}}(y^{1:N}|u_2^{1:N},x_1^{1:N}). \notag
\end{aligned}
\end{eqnarray}
Then the variational distance between $Q_{U_2^{1:N},X_1^{1:N},Y^{1:N}}$ and $P_{U_2^{1:N},X_1^{1:N},Y^{1:N}}$ can be bounded as
\begin{eqnarray}
\begin{aligned}
&2\|Q_{U_2^{1:N},X_1^{1:N},Y^{1:N}}-P_{U_2^{1:N},X_1^{1:N},Y^{1:N}}\|= \sum_{u_2^{1:N},x_1^{1:N},y^{1:N}} |Q(u_2^{1:N},x_1^{1:N},y^{1:N})-P(u_2^{1:N},x_1^{1:N},y^{1:N})|\\  \notag\
&=\sum_{u_2^{1:N},x_1^{1:N},y^{1:N}}|Q(u_2^{1:N}|x_1^{1:N})Q(x_1^{1:N}) Q(y^{1:N}|u_2^{1:N},x_1^{1:N})-P(u_2^{1:N}|x_1^{1:N})P(x_1^{1:N}) P(y^{1:N}|u_2^{1:N},x_1^{1:N})|\\
&\stackrel{(a)}\leq\sum_{u_2^{1:N},x_1^{1:N},y^{1:N}} |Q(u_2^{1:N}|x_1^{1:N})-P(u_2^{1:N}|x_1^{1:N})|P(x_1^{1:N})P(y^{1:N}|u_2^{1:N},x_1^{1:N})\\
&+ \sum_{u_2^{1:N},x_1^{1:N},y^{1:N}} |Q(x_1^{1:N})-P(x_1^{1:N})|Q(u_2^{1:N}|x_1^{1:N})P(y^{1:N}|u_2^{1:N},x_1^{1:N})\\
\end{aligned}
\end{eqnarray}
where inequation $(a)$ follows from \cite[Equation (56)]{aspolarcodes}, $Q(y^{1:N}|u_2^{1:N},x_1^{1:N})=P(y^{1:N}|u_2^{1:N},x_1^{1:N})$.
For the first summation, following the same fashion as the proof of Theorem \ref{theorem:codingtheorem}, we can prove
\begin{eqnarray}
\sum_{u_2^{1:N},x_1^{1:N},y^{1:N}} |Q(u_2^{1:N}|x_1^{1:N})-P(u_2^{1:N}|x_1^{1:N})|P(x_1^{1:N})P(y^{1:N}|u_2^{1:N},x_1^{1:N}) \leq2N\sqrt{4\text{ln}2\cdot2^{-N^\beta}}. \nonumber
\end{eqnarray}
According to the result of the coding scheme for level 1, we already have
\begin{eqnarray}
2\|Q_{U_1^{1:N},Y^{1:N}}-P_{U_1^{1:N},Y^{1:N}}\| \leq2N\sqrt{4\text{ln}2\cdot2^{-N^\beta}}.
\end{eqnarray}
Since we have $P_{Y^{1:N}|U_1^{1:N}}=Q_{Y^{1:N}|U_1^{1:N}}$, we can write
\begin{eqnarray}
2\|Q_{U_1^{1:N}}-P_{U_1^{1:N}}\| \leq2N\sqrt{4\text{ln}2\cdot2^{-N^\beta}}.
\end{eqnarray}
Clearly, there is a one to one mapping between $U_1^{1:N}$ and $X_1^{1:N}$, then we immediately have $2\|Q_{X_1^{1:N}}-P_{X_1^{1:N}}\| \leq2N\sqrt{4\text{ln}2\cdot2^{-N^\beta}}.$ Therefore, for the second summation,
\begin{eqnarray}
\begin{aligned}
\sum_{u_2^{1:N},x_1^{1:N},y^{1:N}} |Q(x_1^{1:N})-P(x_1^{1:N})|&Q(u_2^{1:N}|x_1^{1:N})P(y^{1:N}|u_2^{1:N},x_1^{1:N})\\
&=\sum_{x_1^{1:N}}|Q(x_1^{1:N})-P(x_1^{1:N})|\leq2N\sqrt{4\text{ln}2\cdot2^{-N^\beta}}.
\end{aligned}
\end{eqnarray}
Then we have $||Q_{U_2^{1:N}, X_1^{1:N} ,Y^{1^N}}-P_{U_2^{1:N}, X_1^{1:N} ,Y^{1^N}}||\leq4N\sqrt{4\text{ln}2\cdot2^{-N^\beta}}$, and
\begin{eqnarray}
\begin{aligned}
E[P_e]&=Q_{U_2^{1:N}, X_1^{1:N} ,Y^{1^N}}(\mathcal{E})\\
&\leq \|Q_{U_2^{1:N}, X_1^{1:N} ,Y^{1^N}}-P_{U_2^{1:N}, X_1^{1:N} ,Y^{1^N}}\|+P_{U_2^{1:N}, X_1^{1:N} ,Y^{1^N}}(\mathcal{E})\\
&\leq \|Q_{U_2^{1:N}, X_1^{1:N} ,Y^{1^N}}-P_{U_2^{1:N}, X_1^{1:N} ,Y^{1^N}}\|+\sum_{i\in \mathcal{I}} P_{U_2^{1:N}, X_1^{1:N} ,Y^{1^N}}(\mathcal{E}_i),
\end{aligned}
\end{eqnarray}
The rest part of the proof follows the same fashion of the proof of Theorem \ref{theorem:codingtheorem}. Finally we have $E[P_e]\leq N2^{-N^{\beta'}}$ for any $\beta'<\beta<0.5$.
\end{proof}


%




\bibliographystyle{IEEEtran}
\bibliography{yanfei}

%








\end{document}